\documentclass[a4paper,onecolumn,unpublished,11pt]{quantumarticle}
\pdfoutput=1
\usepackage[utf8]{inputenc}
\usepackage[english]{babel}
\usepackage[T1]{fontenc}
\usepackage{amsmath,dsfont, mathrsfs, amsfonts,amssymb, amsthm, bm}
\usepackage{hyperref}

\usepackage{tikz}
\usepackage{lipsum}

\usepackage[numbers]{natbib}

\usepackage{nicefrac}
\usepackage{comment}
\usepackage{braket}
\usepackage{algorithm}

\DeclareMathAlphabet{\mymathbb}{U}{BOONDOX-ds}{m}{n}

\newcommand{\norm}[1]{\left\lVert#1\right\rVert}
\newcommand{\abs}[1]{\left | #1\right |}

\theoremstyle{definition}
\newtheorem{thm}{Theorem}[section]
\newtheorem{lemma}[thm]{Lemma}

\newtheorem{remark}[thm]{Remark}
\newtheorem{problem}[thm]{Problem}

\newtheorem{proposition}[thm]{Proposition}
\newtheorem{conjecture}[thm]{Conjecture}
\newtheorem{definition}[thm]{Definition}

\begin{document}

\title{Resolvent-based quantum phase estimation: 
Towards estimation of parametrized eigenvalues}

\author[1,2]{Abhijeet Alase}
\email{abhijeet.alase@concordia.ca}
\orcid{0000-0002-5230-4597}
\author[2]{Salini Karuvade}
\orcid{0000-0002-1513-7857}
\affiliation[1]{Centre for Engineered Quantum Systems, School of Physics,
The University of Sydney, Sydney, New South Wales 2006, Australia}
\affiliation[2]{Department of Physics, Concordia University, Montreal, QC H4B 1R6, Canada}
\maketitle
\begin{abstract}
  Quantum algorithms for estimating the eigenvalues of matrices, including the phase estimation algorithm, serve as core subroutines in a wide range of quantum algorithms, including those in quantum chemistry and quantum machine learning. The standard quantum eigenvalue (phase) estimation algorithm accepts 
  a Hermitian (unitary) matrix and a state in an unknown superposition of its eigenstates as input, and 
  coherently records the estimates for real eigenvalues (eigenphases) in an ancillary register.
  Extension of quantum eigenvalue and phase estimation algorithms to the case of non-normal input matrices 
  is obstructed by several factors such as non-orthogonality of eigenvectors, existence of generalized eigenvectors and the fact that eigenvalues may lie anywhere on the complex plane.


In this work, we propose a novel approach for estimating the eigenvalues of non-normal matrices based on 
preparation of a state that we call the ``resolvent state''. We construct the first efficient algorithm for estimating the phases of the unimodular eigenvalues of a given non-unitary matrix. We then construct an efficient algorithm for estimating the real eigenvalues of a given non-Hermitian matrix, achieving complexities that match the best known results while operating under significantly relaxed assumptions on the non-real part of the spectrum. The resolvent-based approach that we introduce also extends to estimating eigenvalues that lie on a parametrized complex curve, subject to explicitly stated conditions, thereby paving the way for a new paradigm of parametric eigenvalue estimation.

\end{abstract}

\section{Introduction}\label{sec:intro}
\paragraph{Context:} 
Quantum algorithms for estimating eigenvalues of a Hermitian matrix and 
eigenphases of a unitary matrix~\cite{Kit95,NC00} underpin solutions to some of the 
most exciting computational problems and promising
applications of quantum computing~\cite{Sho94,BHMT02,HHL09,NWZ09,Ral20,DL23,WBL12,LMR14}. 
In quantum phase estimation (QPE), 
we are given oracle access to a unitary matrix $U$ and copies of a state $\ket{\psi}$
formally expressed as $\ket{\psi}=\sum_l\beta_l\ket{\lambda_l}$, where $\{\ket{\lambda_l}\}$
are the eigenvectors of $U$ with corresponding eigenvalues $\{\lambda_l\}$ that are unimodular 
(i.e. $\abs{\lambda_l}=1$).
The goal is to generate a state of the form 
$\ket{\psi_U} = \sum_l \beta_l \ket{\phi_{\lambda_l}}\ket{\lambda_l}$, 
where each $\ket{\phi_{\lambda_l}}$ has high overlap on the basis states 
corresponding to the binary representation of $\arg(\lambda_l)$ to additive precision $\epsilon$. 
The task of quantum eigenvalue estimation (QEVE) is quite similar, except that 
the input matrix is sparse and Hermitian, and $\ket{\phi_{\lambda_l}}$ encodes the binary representation of $\lambda_l$ scaled by a suitable constant. QPE and QEVE are closely related;
in fact, several approaches to QEVE use QPE as a subroutine~\cite{KOS07,Ral20}.

\paragraph{Background and knowledge gap:}
Whereas the first efficient algorithm for quantum phase estimation (QPE) of unitary matrices 
was developed almost three decades ago~\cite{Kit95}, its extension to sparse 
or block-encoded non-unitary
matrices has remained an unsolved problem. The algorithms or techniques proposed
in the literature lead to complexities exponential in at 
least one of the input parameters~\cite{WWLN10,DGK14,DK17}. No efficient algorithm for estimating phases of eigenvalues is known even for the special case in which 
all eigenvalues of the given matrix lie on the unit circle. 

Quantum algorithms for estimating eigenvalues of sparse Hermitian matrices 
were developed alongside QPE, but their efficient extensions to
non-Hermitian matrices have been constructed only in the last three years~\cite{Sha22,SL22,LS24}. The emergence of new linear algebraic techniques for non-Hermitian matrices~\cite{Ber14,TOSU21,TOSU22,LS24} based on quantum linear system algorithms (QLSAs)~\cite{HHL09,CAS+22,Kro23} has enabled efficient algorithms for QEVE in different settings~\cite{Sha22,SL22,LS24,zhang2024exponential,ZZHY24}. In particular, two approaches, one based on ``Fourier state'' generation~\cite{Sha22,SL22} and another based on ``Chebyshev history state'' generation~\cite{LS24} have been  proposed for QEVE of non-Hermitian matrices with only real eigenvalues. The non-normality of these matrices manifests in the form of 
non-orthogonal eigenvectors and the presence of generalized eigenvectors. 
The limitation of both these approaches lies in their assumption, namely that the spectrum of the input matrix $A$ is real. The underlying technique used by both these approaches does not lend itself to a straightforward generalization to the case where $A$ has non-real eigenvalues. 

\begin{figure}
    \centering
    \includegraphics[height=0.4\linewidth]{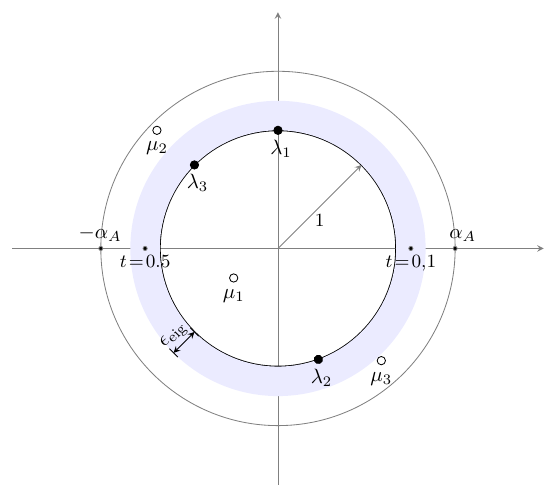}
    \includegraphics[height=0.4\linewidth]{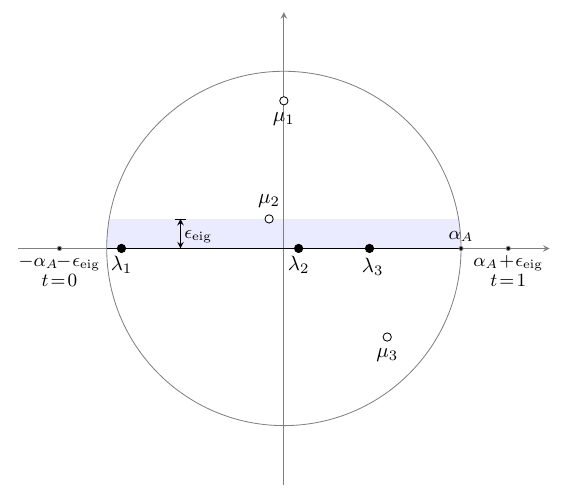}
    \caption{The location of eigenvalues in the problem statement for (left) QEUE and (right) QERE. Eigenvalues are promised to not lie inside the shaded region.}
    \label{fig:setting}
\end{figure}

\paragraph{Results:}
In this work, we consider extensions of QPE and QEVE to non-normal matrices with relaxed assumptions on the spectrum. We call these extensions Quantum coherent Estimation of unimodular Eigenvalues (QEUE) and Quantum coherent Estimation of Real Eigenvalues (QERE) respectively. Here the word ``coherent'' 
refers to coherently recording the answer (eigenvalue or eigenphase) on an ancilla register,
as opposed to returning a classical estimate of the same.
We construct very simple and efficient quantum algorithms to solve these problems. 

We state our problems using the block encoding framework. 
Recall that a unitary $O_A$ is an $\alpha_A$-block encoding of a matrix $A$ if 
$\alpha_A\braket{0^{a_A}|O_A|0^{a_A}} = A$, where $\ket{0^{a_A}}$ is the fiducial state of
the ancilla qubits. We assume that the block encoding of the input matrix $A$ is perfect.
Let us now state the problem of QEUE, which generalizes QPE to non-unitary matrices (see Fig.~\ref{fig:setting}). 
\begin{problem}[Informal statement of QEUE (Problem~\ref{prob:qpe})]
\label{prob:inf_qpe}
We are given an $\alpha_A$-block encoding $O_A$ of an $n$-qubit matrix $A$, two accuracy parameters $\epsilon_{\rm st}, \epsilon_{\rm eig}$, and a unitary $O_\psi$ that generates a state $\ket{\psi} = \sum_l \beta_l\ket{\lambda_l}$ such that each $\ket{\lambda_l}$ is an eigenstate of $A$ with unimodular eigenvalue $\lambda_l$, i.e. $\abs{\lambda_l}=1$. All other eigenvalues $\{\mu_k\}$ of $A$ have magnitude outside the interval $(1,1+\epsilon_{\rm eig})$. The goal is to generate an $(n+a)$-qubit quantum state 
\begin{equation}
    \ket{\psi_A} = \frac{\sum_l \beta_l \ket{\phi_{\lambda_l}}\ket{\lambda_l}}
    {\norm{\sum_l \beta_l \ket{\phi_{\lambda_l}}\ket{\lambda_l}}}
\end{equation}
with accuracy $\epsilon_{\rm st}$, where each $\ket{\phi_{\lambda_l}}$ is a state that, when measured in the computational basis, yields an estimate of $\operatorname{arg}(\lambda_l)$ to additive accuracy $\epsilon_{\rm eig}$ with unit probability.
\end{problem}
Note that in QEUE, only the eigenvalues to be estimated are required to be on the unit circle, 
and other eigenvalues could take any complex values except in an annulus defined by 
radii $1$ and $1+\epsilon_{\rm eig}$ (see Fig.~\ref{fig:setting}).
Our algorithm for QEUE achieves the complexity stated in the following theorem.
\begin{thm}[Informal statement of Theorem~\ref{thm:qpeideal}]
\label{thm:inf_qpa}
Problem~\ref{prob:inf_qpe} can be solved by
    \begin{equation}
        \mathbf{O}\left(\frac{ \alpha_A \kappa_S^4\mathcal{K}^\circ_\delta(A)}{\epsilon_{\rm eig}\epsilon_{\rm st}^2} 
        \log\left(\frac{1}{\epsilon_{\rm st}}\right)\right), \quad 
        \mathcal{K}_\delta^\circ(A)=\delta \sup_{|z| = 1+\delta}\|(z I-A)^{-1}\|
    \end{equation}
uses of $O_A$ and $O_\psi$, where $\delta \in {\bf \Theta} (\epsilon_{\rm eig}\epsilon_{\rm st}^2/\kappa_S^4)$. 
\end{thm}
Here $\kappa_S= \sigma_{\max}(S)/\sigma_{\min}(S)$ is the 
condition number of the (rectangular) 
matrix $S$ whose columns are $\{\ket{\lambda_l}\}$, and $\sigma_{\max}(S)$ ($\sigma_{\min}(S)$)
denotes the highest (lowest non-zero) singular value of $S$.
The quantity $\mathcal{K}_\delta^\circ(A)$ bounds from below the 
Kreiss constant of A with respect to the unit circle, 
which we denote by $\mathcal{K}^{\circ}(A)$ (see Eq.~\eqref{eq:kreisscirc}).
The Kreiss constant $\mathcal{K}^{\circ}(A)$ is a lower bound on the
maximum transient growth $\sup_{k>0}\norm{A^k}$ of the discrete dynamics generated by $A$ for a stable $A$~\cite{TE05}. However, $\mathcal{K}_\delta^\circ(A)$
appearing in our complexity can be finite even if the transient growth is unbounded.
For a diagonalizable $A$, we have $\mathcal{K}_\delta^\circ(A) \in {\bf O}(\bar{\kappa}_A)$, 
where $\bar{\kappa}_A$ is the condition number of the similarity transformation that diagonalizes $A$
(see Def.~\ref{def:jordcond} for a rigorous definition).
This leads to the overall complexity 
${\bf O}\left(\alpha_A \bar{\kappa}_A\kappa_S^4/\epsilon_{\rm eig}\epsilon_{\rm st}^2 \log(1/\epsilon_{\rm st})\right)$. A similar bound for complexity is derived for non-diagonalizable $A$ in terms of the dimension of its largest Jordan block.

The problem QERE, stated below informally, extends QEVE to non-normal 
matrices with possibly non-real eigenvalues.
\begin{problem}[Informal statement of QERE (Problem~\ref{prob:qeve})]
\label{prob:inf_qeve}
We are given an $\alpha_A$-block encoding $O_A$ of an $n$-qubit matrix $A$, two accuracy parameters $\epsilon_{\rm st}, \epsilon_{\rm eig}$, and a unitary $O_\psi$ that generates a state $\ket{\psi} = \sum_l \beta_l\ket{\lambda_l}$ such that each $\ket{\lambda_l}$ is an eigenstate of $A$ with a real eigenvalue $\lambda_l$. All other eigenvalues $\{\mu_k\}$
of $A$ have imaginary parts lie outside the interval $(0,\epsilon_{\rm eig})$. The goal is to generate an $(n+a)$-qubit quantum state 
\begin{equation}
    \ket{\psi_A} = \frac{\sum_l \beta_l \ket{\phi_{\lambda_l}}\ket{\lambda_l}}
    {\norm{\sum_l \beta_l \ket{\phi_{\lambda_l}}\ket{\lambda_l}}}
\end{equation}
with accuracy $\epsilon_{\rm st}$, where each $\ket{\phi_{\lambda_l}}$ is a state that, when measured in the computational basis, yields a scaled estimate of $\lambda_l$ to additive accuracy $\epsilon_{\rm eig}$ with unit probability.
\end{problem}

Our algorithm achieves the complexity stated in the following theorem.
\begin{thm}[Informal statement of Theorem~\ref{thm:qeveideal}]
\label{thm:inf_qeveideal}
Problem~\ref{prob:inf_qeve} can be solved by
    \begin{equation}
        \mathbf{O}\left(\frac{ \alpha_A \kappa_S^4\mathcal{K}_\delta(-iA)}{\epsilon_{\rm eig}\epsilon_{\rm st}^2} 
        \log\left(\frac{1}{\epsilon_{\rm st}}\right)\right),\quad
        \mathcal{K}_\delta(-iA)=\delta \sup_{\Re(z)=\delta}\|(z\mathds{1}+iA)^{-1}\|\,.
    \end{equation}
uses of $O_A$ and $O_\psi$, 
where $\delta \in \Theta(\epsilon_{\rm eig}\epsilon_{\rm st}^2/\kappa_S^4)$.
\end{thm}
The quantity $\mathcal{K}_\delta(-iA)$ is a lower bound on the Kreiss constant of 
$-iA$ with respect to lower half plane, which we denote by $\mathcal{K}(-iA)$. 
The Kreiss constant $\mathcal{K}(-iA)$ provides a lower bound 
on the maximum transient growth $\sup_{t>0}\norm{e^{-iAt}}$ for the continuous dynamics generated by $A$, provided $A$ is stable. 
For diagonalizable $A$, we have $\mathcal{K}_\delta(-iA) \in {\bf O}(\bar{\kappa}_A)$, which leads to the overall complexity
$\mathbf{O}\left( \alpha_A \bar{\kappa}_A\kappa_S^4/\epsilon_{\rm eig}\epsilon_{\rm st}^2 \log(1/\epsilon_{\rm st}\right)$ for QERE for this case. 

If all eigenvalues of $A$ are real, 
then QERE reduces to the coherent version of QEVE, which we discuss
ahead in \S~\ref{sec:bgpreviouswork}. 
For constant $\kappa_S$ and $\epsilon_{\rm st}$, 
 the complexity of our
algorithm for QEVE becomes 
$\mathbf{O}\left( \alpha_A \mathcal{K}_\delta(-iA)/\epsilon_{\rm eig}\right)$.
If $A$ is diagonalizable, then the complexity of our algorithm is 
$\mathbf{O}\left( \alpha_A \bar{\kappa}_A/\epsilon_{\rm eig}\right)$, which 
matches the previous best results~\cite{LS24}. If $A$ is not diagonalizable, i.e.
if the dimension of the largest block is $d\ge 2$, then our algorithm
scaling as $\mathbf{O}\left( \alpha_A \bar{\kappa}_A/\epsilon_{\rm eig}^d\right)$
has better scaling in terms of $\alpha_A$ than the previous best result~\cite{LS24}.

\paragraph{New Techniques:}
Our algorithms for both QEUE and QERE are conceptually very simple
and rely on a single application of QLSA. 
The algorithm for QEUE begins by choosing a set of candidate points 
$\{x_j:=e^{i\phi_j}\}$ to approximate the eigenvalues
in $\{\lambda_l\}$.
These values are chosen to be spaced uniformly on the unit circle. The main idea is to prepare a state proportional to 
\begin{equation}
\label{eq:qpestate}
    \ket{\tilde{\psi}_{A}} \propto \sum_l \beta_l \left(\sum_{j} f(\lambda_l,x_j)\ket{j}\right)\ket{\lambda_l},
\end{equation}
where $\{\ket{j}\}$ are the computational basis states of an ancilla register.
For a fixed $\lambda_l$, we need to ensure that the states $\{\ket{j}: x_j\approx\lambda_l\}$ have higher amplitudes.
Therefore, we choose $f(\lambda_l,x_j) = (x_j(1+\delta)-\lambda_l)^{-1}$. 
This choice also ensures that it is easy to prepare the state $\ket{\tilde{\psi}_{A}}$
by noting that it is the solution of 
the linear equation $M\bm{x} = \bm{b}$, where 
\begin{equation}
    M = Z\otimes \mathds{1}_n - \mathds{1}_a\otimes A,
    \quad Z = \sum_j z_j \ket{j}\bra{j},\quad z_j = x_j(1+\delta),\quad {\rm and}\,\, \bm{b} = \ket{+^a}\ket{\psi},
\end{equation}
and $\ket{+^a}$ is the equal superposition state of all computational basis states 
of the ancillary register.
A block encoding of $M$ is easy to 
generate using standard techniques. Note that $M^{-1}$ is a block diagonal
matrix with $j$th block $M_j^{-1} = (z_j-A)^{-1} = :R(z_j)$ being the resolvent of
$A$ evaluated at $z_j$. Therefore, for the particular choice of  $f(\lambda_l,x_j)$,
we call the state $\ket{\tilde{\psi}_A}$ in Eq.~\eqref{eq:qpestate} 
the ``resolvent state''. 

The correctness of our algorithm relies on the proof that the probability\\ 
$\sum_{|\phi_j-\arg(\lambda_l)|\le \epsilon_{\rm eig}} |f(\lambda_l,x_j)|^2$, 
corresponding to the ancilla register carrying an $\epsilon_{\rm eig}$-additive parametric 
approximation to $\lambda_l$, dominates the complementary probability\\
$\sum_{|\phi_j-\arg(\lambda_l)| > \epsilon_{\rm eig}} |f(\lambda_l,x_j)|^2$. We prove this by relating these discrete sums to integrals, while carefully accounting for the discretization errors. 
The failure probability is then bounded by explicit integration, which reveals the suitable
choice of parameters to achieve accuracy $\epsilon_{\rm st}$.
The complexity of the algorithm follows from that of QLSA.

The algorithm for QERE is very similar to that for QEUE. In this case, the candidate
approximations $\{x_j \}$ to $\{\lambda_l\}$ 
are chosen to be spaced uniformly on a real interval.
The algorithm prepares the resolvent state similar to the one in Eq.~\eqref{eq:qpestate} 
by an application of QLSA,
but in this case we use $f(\lambda_l,x_j) = (x_j+i\delta-\lambda_l)^{-1}$  and $z_j = x_j+i\delta$.
The error and complexity analysis is similar 
to that of QEUE. 

\paragraph{Discussion:} 
The resolvent state-based approach introduced in this work treats QEUE and QERE problems on an equal footing, producing remarkably similar, extremely simple and efficient algorithms for the two tasks. Moreover, our algorithm extends to the case when the eigenvalues to be estimated are neither real nor unimodular, and instead lie on a known parametrized complex curve satisfying some conditions that we state in \S~\ref{sec:peve}. 
Thus, our work establishes the problem of estimating parametrized eigenvalues as a natural extension and unification of QPE and QEVE, and provides a promising technique for its solution. These successes pave the way for a new paradigm of parametric eigenvalue estimation as a general approach to estimating eigenvalues of non-normal matrices.  

Our algorithms for QEUE and QERE apply to an important class of operators studied in
non-Hermitian physics~\cite{BB98,AKS22,Kar22}. Specifically, our algorithm for QERE can be applied to
the estimation of eigenvalues of an unbroken PT-symmetric Hamiltonian $H$; such
Hamiltonians are known to have real spectrum~\cite{Mos02,KAS22}. It can also be applied to PT-symmetric Hamiltonians
with broken PT symmetry, which have some real eigenvalues and remaining eigenvalues
in complex conjugate pairs~\cite{BBJ03}. In this latter case, the state $\ket{\psi}$ used for eigenvalue estimation 
should have support only on PT-unbroken eigenvectors for our algorithm to apply.
Similarly, our algorithm for QEUE can be applied for estimation of eigenphases of
a time-propagator $U$ generated by an unbroken PT-symmetric Hamiltonian. Such propagators 
are known to have unimodular eigenvalues. Our algorithm for QEUE would also apply
to a broken PT-symmetric propagator, assuming the state $\ket{\psi}$ used for the estimation
has support only on the PT-unbroken eigenvectors. Both the algorithms for QEUE and QERE
could be helpful in estimation of ground energies for PT-symmetric systems.

Our algorithm for QEUE is the first efficient algorithm that successfully extends QPE to non-unitary matrices. 
In \S~\ref{sec:bgpreviouswork}, we discuss the challenges in extending known 
approaches for QPE and QEVE to QEUE.
For QEVE, the algorithm based on Chebyshev history state generation~\cite{LS24} 
was the state-of-the-art before this work.
If $A$ is diagonalizable, then the complexity of our algorithm for QERE
matches this algorithm (see \S~\ref{sec:bgpreviouswork} for more details).
If $A$ is not diagonalizable, the complexity of our algorithm  has better dependence on $\alpha_A$  while matching ~\cite{LS24} in all other input parameters.
As mentioned before, both 
Chebyshev-history state based and Fourier-state based~\cite{Sha22,SL22} algorithms for QEVE assume
that all eigenvalues of $A$ are real. Their techniques
do not extend to the case of complex eigenvalues in a straightforward manner,
particularly when at least one eigenvalue of $A$ has positive imaginary part. In such a case, 
$A$ generates unstable dynamics, leading to an exponential query cost for a
Fourier-state based approach for QERE.  
The approach based on Chebyshev history state generation is limited
by the condition number of the matrix associated to the generator of
the Chebyshev functions, which grows exponentially when $A$ has complex eigenvalues.

Regarding the techniques used, a state similar to the resolvent state in Eq.~\eqref{eq:qpestate} was used for matrix-function implementation using  the Cauchy integral approach~\cite{TOSU22}. The utility of this state for phase estimation was not recognized; in contrast, this approach was thought to circumvent the (standard) QPE. Our work establishes that this approach indeed performs imperfect QPE, which surprisingly allows for efficient function implementation, countering the usual notion that QPE-based approaches are inefficient for this task.

\paragraph{Organization:}
The organization of the rest of the paper is as follows. In \S~\ref{sec:bg}, we review some definitions and basic results on discretization of curves and integrals, Jordan form decomposition of matrices, the block encoding framework for matrix arithmetic and previous results on quantum phase and eigenvalue estimation. In \S~\ref{sec:qeve}, we formalize the problem
of estimating real eigenvalues of a given non-normal matrix, and construct the algorithm based on
the resolvent approach. In \S~\ref{sec:qpe}, we formalize an analogous problem for estimating the unimodular eigenvalues of a non-normal matrix and construct a resolvent-based algorithm for the same. Finally in \S~\ref{sec:peve}, we discuss extension of our approach to estimation of eigenvalues on a parametrized curve in the complex plane.

\section{Background and preliminaries}
\label{sec:bg}
In this section, we review the concepts required for developing the quantum algorithms discussed in this paper. We also include some technical lemmas that we use in later sections.

\subsection{Discretization of a curve and Riemann sum}
We start this section by reviewing the definitions of a curve in a complex plane, discretization of a curve and Riemann sum. We also provide a bound on the error obtained when approximating an integral on a curve with a discrete Riemann sum.\begin{definition}[Complex curves and their discretization]
    A continuous function $\gamma:[0,1] \to \mathbb{C}$ is called a {\it curve} in a complex plane. 
    A curve $\gamma$ is said to be {\it closed} if $\gamma(0) = \gamma(1)$. For an $N \in \mathbb{Z}^+$, an {\it $N$-discretization} of $\gamma$ is the ordered set
\begin{equation}
\label{eq:discrete_curve}
        [\gamma,N] := \left(\gamma(t_j)\right)_{j=0}^{N-1}, \quad 
        t_j := \frac{j}{N}.
\end{equation} 
An ordered set $\Gamma\in [0,1]^N$ of $N$ real numbers is called a {\it discrete curve} if it is an $N$-discretization of some complex curve $\gamma$.
\end{definition}
For a closed curve $\gamma$, it is sometimes easier to consider the extension of $\gamma$ to
$\tilde{\gamma}:\mathbb{R}\to \mathbb{C}$ such that $\tilde{\gamma}(t) = \gamma(t \mod 1)$.
In other words, we treat $t$ and $j$ as modular variables with 
\begin{equation}
    t = t\mod 1, \quad j = j \mod N.
\end{equation}
Such a treatment simplifies presentation of some proofs in \S~\ref{sec:qpe}.

\begin{definition}[Riemann sum]
\label{def:riemannsum}
For a continuous function $f:[0,1]\to \mathbb{C}$, an integer $N \in \mathbb{Z}^+$ and a closed interval 
$[t_{\rm min}, t_{\rm max}] \subseteq [0,1]$, 
the {\it Riemann sum} is defined as
\begin{equation}
        \sum_{t_{\rm min}}^{t_{\rm max}}[f,N]:= \frac{1}{N}\sum_{j = \lceil Nt_{\rm min} \rceil}^{\lfloor Nt_{\rm max}\rfloor}f(j/N).
\end{equation}

\end{definition}
\begin{lemma}[Error in Riemann sum]
\label{lem:riemannsum}
    For a differentiable function $f:(0,1)\to \mathbb{C}$, an integer $N\mathbb{Z}^+$ and a closed interval 
$[t_{\rm min}, t_{\rm max}] \subseteq [0,1]$, we have
\begin{equation}
        \abs{\int_{t_{\rm min}}^{t_{\rm max}} f(t)dt - \sum_{t_{\rm min}}^{t_{\rm max}}[f,N]} \le 
        \frac{1}{N}\left\{\frac{(t_{\rm max}-t_{\rm min})^2}{2}
        \left(\max_{t\in(t_{\rm min},t_{\rm max})}\abs{\frac{df}{dt}}\right) + 2\max_{t\in(t_{\rm min},t_{\rm max})}
        \abs{f(t)}\right\}.
\end{equation}
\end{lemma}

\begin{proof}
First note that 
\begin{equation}
\abs{\int_{t_{\rm min}}^{t_{\rm max}} f(t)dt - \int_{(\lceil Nt_{\rm min} \rceil)/N}^{(\lfloor Nt_{\rm max}\rfloor+1)/N} f(t)dt}
\le \frac{2\max_{t\in(t_{\rm min},t_{\rm max}})\abs{f(t)}}{N}.
\end{equation}
It is known that~\cite{hughes2020calculus}
\begin{equation}
\abs{\int_{(\lceil Nt_{\rm min} \rceil)/N}^{(\lfloor Nt_{\rm max}\rfloor+1)/N} f(t)dt - \sum_{t_{\rm min}}^{t_{\rm max}}[f,N]} \le 
       \frac{(t_{\rm max}-t_{\rm min})^2}{2N}\left(\max_{t\in(t_{\rm min},t_{\rm max})}\abs{\frac{df}{dt}}\right).
\end{equation}
The statement of the Lemma now follows by triangle inequality.
\end{proof}

\subsection{Jordan form decomposition and Kreiss constants for a matrix}

In this section, we review the Jordan form decomposition of matrices and review some useful results on the Kreiss constants. We closely follow the presentation of former topic in Ref.~\cite{LS24}.

We denote the space of complex $N\times N$ matrices by $\mathbb{C}^{N \times N}$. The spectrum of a matrix $C \in \mathbb{C}^{N \times N}$ is denoted by $\Lambda(C)$. 

\begin{proposition}[Jordan form decomposition]
    The following statements hold for a matrix $C\in\mathbb{C}^{N\times N}$:
    There exist an invertible matrix $T\in\mathbb{C}^{N\times N}$ and a block diagonal matrix 
        \begin{equation}
        \label{eq:jordan_block}
        \begin{aligned}
            J&=
            \begin{bmatrix}
                J(\lambda_0,d_0) & & &\\
                & J(\lambda_1,d_1) & &\\
                & &\ddots &\\
                & & &J(\lambda_{s-1},d_{s-1})
            \end{bmatrix}\in\mathbb{C}^{N\times N},\\
            J(\lambda_l,d_l)&=
            \begin{bmatrix}
                \lambda_l &  &  &  &  & \\
                1 & \lambda_l &  &  &  & \\
                0 & 1 & \lambda_l &  &  & \\
                \vdots & 0 & 1 & \ddots &  & \\
                \vdots & \ddots & \ddots & \ddots & \lambda_l & \\
                0 & \cdots & \cdots & 0 & 1 & \lambda_l\\
            \end{bmatrix}\in\mathbb{C}^{d_l\times d_l},
        \end{aligned}
        \end{equation}
        such that $C=TJT^{-1}$.
\end{proposition}

\begin{proposition}[Jordan form decomposition of a function of a matrix]
Given a matrix $C$ with Jordan form decomposition $C=TJT^{-1}$, 
if a scalar function $f(x)$ is analytic at all eigenvalues of $C$, then
\begin{equation}
\label{eq:jordan_form_transformation}
\begin{aligned}
    f(C)&=Tf(J)T^{-1}
    =T
    \begin{bmatrix}
        f(J(\lambda_0,d_0)) & & &\\
        & f(J(\lambda_1,d_1)) & &\\
        & &\ddots &\\
        & & &f(J(\lambda_{s-1},d_{s-1}))     
    \end{bmatrix}T^{-1},\\
    f(J(\lambda_l,d_l))&=
    \begin{bmatrix}
        f(\lambda_l) &  &  &  &  & \\
        f'(\lambda_l) & f(\lambda_l) &  &  &  & \\
        \frac{f^{(2)}(\lambda_l)}{2!} & f'(\lambda_l) & f(\lambda_l) &  &  & \\
        \vdots & \frac{f^{(2)}(\lambda_l)}{2!} & f'(\lambda_l) & \ddots &\\
        \vdots & \ddots & \ddots & \ddots & \ddots & \\
        \frac{f^{(d_l-1)}(\lambda_l)}{(d_l-1)!} & \cdots & \cdots & \frac{f^{(2)}(\lambda_l)}{2!} & f'(\lambda_l) & f(\lambda_l)\\
    \end{bmatrix}.
\end{aligned}
\end{equation}
In particular, for $f(C) = (z\mathds{1}-C)^{-1}$ for some $z\notin \Lambda(C)$, 
\begin{multline}
    f(J(\lambda_l,d_l))=(z\mathds{1}-J(\lambda_l,d_l))^{-1}\\
    =
    \begin{bmatrix}
        (z-\lambda_l)^{-1} &  &  &  &  & \\
        (z-\lambda_l)^{-2} & (z-\lambda_l)^{-1} &  &  &  & \\
        (z-\lambda_l)^{-3} & (z-\lambda_l)^{-2} & (z-\lambda_l)^{-1} &  &  & \\
        \vdots & (z-\lambda_l)^{-3} & (z-\lambda_l)^{-2} & \ddots &\\
        \vdots & \ddots & \ddots & \ddots & \ddots & \\
        (z-\lambda_l)^{-d_l} & \cdots & \cdots & (z-\lambda_l)^{-3} & (z-\lambda_l)^{-2} & 
        (z-\lambda_l)^{-1}\\
    \end{bmatrix}.
\end{multline}
\end{proposition}

\begin{definition}[Jordan Condition Number of a Matrix]
    \label{def:jordcond}
    the Jordan condition number $\bar{\kappa}_C$ of $C\in\mathbb{C}^{N\times N}$ is
    \begin{equation}
        \bar{\kappa}_C := \min_{T \in \mathcal{T}}\norm{T}\norm{T^{-1}},
    \end{equation}
    where 
    \begin{equation}
    \mathcal{T} = \{T: C = TJT^{-1} \text{ is a Jordan decomposition of $C$}\}.
    \end{equation}
\end{definition}
\begin{lemma}[Spectral norm bound for matrices]
\label{lem:block_norm}
For any  $C\in\mathbb{C}^{N\times N}$,
    \begin{equation}
        \norm{C} \leq
        \sqrt{\max_{1\leq k\leq N}\sum_{j=1}^N|C_{jk}|}
        \sqrt{\max_{1\leq j\leq N}\sum_{k=1}^N|C_{jk}|}.
    \end{equation}
\end{lemma}

\begin{proposition}[Spectral norm bound on the resolvent]
\label{lem:resolventbound}
For a matrix $C$ with Jordan form decomposition $C=TJT^{-1}$, let $d = \max_l\{d_l\}$ be the dimension of the largest Jordan block. Let $z \notin \Lambda(C)$ and $\delta = \min_l|z-\lambda_l|$ be the distance
between $z$ and the eigenvalue $\lambda_l$ of $C$. Suppose $\delta < 1$. Then  
\begin{equation}
    \norm{(z\mathds{1}-C)^{-1}} \le  \frac{\bar{\kappa}_C (1-\delta^d)}{\delta^d(1-\delta)}.
\end{equation}
\end{proposition}
\begin{proof}
Observe that 
\begin{equation}
        \norm{(z\mathds{1}-C)^{-1}} = \norm{T(z\mathds{1}-J)^{-1}T^{-1}}
        \le \bar{\kappa}_C \norm{(z\mathds{1}-J)^{-1}}.
    \end{equation}
By Lemma~\ref{lem:block_norm}, we have $\norm{(z\mathds{1}-J)^{-1}} 
    \le \delta^{-d}(1-\delta^d)/(1-\delta)$,
which completes the proof.
\end{proof}

We now discuss two kinds of Kreiss constants, related to the transient growth of the continuous-time and discrete-time dynamics generated by a matrix, respectively. 
\begin{proposition}[Kreiss matrix theorem \cite{TE05}]
 For an arbitrary $C\in\mathbb{C}^{N\times N}$, the following bounds are optimal:
 \begin{enumerate}
     \item 
     \begin{equation}
     \mathcal{K}(C)\leq \sup_{t\geq 0}\|\exp(Ct)\|\leq eN\mathcal{K}(C)\,,
 \end{equation}
 where $e$ is the Euler constant and 
 the Kreiss constant with respect to left-half plane, $\mathcal{K}(C)$, is defined as
 \begin{equation}
 \label{eq:kreisslin}
     \mathcal{K}(C)=\sup_{\Re(z)>0}\Re(z)\|(z I-C)^{-1}\|\,.
 \end{equation}

 \item 
 \begin{equation}
     \mathcal{K}^\circ(C)\leq \sup_{k\geq 0}\|C^k\|\leq eN\mathcal{K}^\circ(C)\,,
 \end{equation}
 where $e$ is the Euler constant and 
 the Kreiss constant with respect to the unit circle, $\mathcal{K}^\circ(C)$, is defined as
 \begin{equation}
 \label{eq:kreisscirc}
     \mathcal{K}^\circ(C)=\sup_{|z|>1}(|z|-1)\|(z I-C)^{-1}\|\,.
 \end{equation}
 \end{enumerate}
 
\end{proposition}

\begin{definition}[Restricted Kreiss constants]
\label{def:Kdelta}\label{def:Kdeltacirc}
For an arbitrary $C\in\mathbb{C}^{N\times N}$ and a $\delta >0$, define
\begin{equation}
     \mathcal{K}_\delta(C)=\delta \sup_{y\in\mathbb{R}}\|((\delta + iy) I-C)^{-1}\|\,.
 \end{equation}
and 
\begin{equation}
     \mathcal{K}_\delta^\circ(C)=\delta \sup_{|z| = 1+\delta}\|(z I-C)^{-1}\|\,.
 \end{equation}
\end{definition}
Notice that $\mathcal{K}(C) = \sup_{\delta > 0}\mathcal{K}_\delta(C)$, and therefore 
$\mathcal{K}_\delta(C) \le \mathcal{K}(C)$. Moreover, $\mathcal{K}_\delta(C)$ can be finite
for some values of $\delta$ even if $\mathcal{K}(C)$ is unbounded, as is the case for 
any matrix $C$ with a defective imaginary eigenvalue. Similarly,
$\mathcal{K}^\circ(C) = \sup_{\delta > 0}\mathcal{K}_\delta^\circ(C)$, and therefore 
$\mathcal{K}_\delta^\circ(C) \le \mathcal{K}^\circ(C)$. 
Moreover, $\mathcal{K}_\delta^\circ(C)$ can be finite
for some values of $\delta$ even if $\mathcal{K}^\circ(C)$ is unbounded, as is the case for 
any matrix $C$ with a defective eigenvalue on the unit circle.

\begin{lemma}[Spectral norm bound on restricted Kreiss constants]
\label{prop:kreissbound}
For a matrix $C$ with Jordan form decomposition $C=SJS^{-1}$, let $d = \max_l\{d_l\}$ be the dimension of the largest Jordan block.
\begin{enumerate}
    \item If $\Re(\lambda_l) \notin (0,2\delta)$ $\forall l$ and for some $\delta \in (0,1)$.
Then  
\begin{equation}
    \mathcal{K}_\delta(C) \le  \bar{\kappa}_C \left(\frac{1}{\delta}\right)^{d-1}\left(\frac{1-\delta^d}{1-\delta}\right).
\end{equation}
    \item If $|\lambda_l| \notin (1,1+2\delta)$ $\forall l$ for some $\delta \in (0,1)$.
    Then  
    \begin{equation}
        \mathcal{K}_\delta^\circ(C) \le  \bar{\kappa}_C \left(\frac{1}{\delta}\right)^{d-1}\left(\frac{1-\delta^d}{1-\delta}\right).
    \end{equation}
\end{enumerate}

\end{lemma}
\begin{proof}
\begin{enumerate}
    \item Recall that $\mathcal{K}_\delta(C)=\delta \sup_{y\in\mathbb{R}}\|((\delta + iy) I-C)^{-1}\|$. Since $\Re(\lambda_l) \notin (0,2\delta)$, we have $|(\delta + iy)-\lambda_l| \ge \delta$ $\forall y \in \mathbb{R}$. From Lemma~\ref{lem:resolventbound}, $\norm{\left((\delta + iy)\mathds{1}-C\right)^{-1}} 
\le  \bar{\kappa}_C \delta^{-d}(1-\delta^d)/(1-\delta)$ $\forall y \in \mathbb{R}$, which leads to the desired bound on $\mathcal{K}_\delta(C)$.
    \item The proof is very similar to the first part. We use the fact that any $z$ with $|z| = 1+\delta$ satisfies $|z-\lambda_l| \ge \delta$ $\forall l$. The desired bound follows from Lemma~\ref{lem:resolventbound}. 
\end{enumerate}
\end{proof}

\subsection{Matrix arithmetic in the block encoding framework}
In this section, we review the relevant results from the literature on block encoding of matrices.
\begin{definition}[Block encoding of a matrix \cite{GSLW19}]
\label{def:block encoding}
An $(n+a_C)$-qubit unitary operator $O_C$ is an $\alpha_C$-block encoding of $C\in \mathbb{C}^{2^n\times 2^n}$ 
for some $\alpha_C \in \mathbb{R}^+$ if 
\begin{equation}
    C/\alpha_C  =  \bra{0^{a_C}}O_C \ket{0^{a_C}}, 
\end{equation}
where $ \bra{0^{a_C}} O_C \ket{0^{a_C}}  := \sum_{i,j \in [2^n]}\left(\bra{0^{a_C}}\bra{i} O_C \ket{0^{a_C}}\ket{j} 
\right) \ket{i}\bra{j}$.
\end{definition}

\begin{lemma}[Linear combination of block-encoded matrices, Lemma~29 in~\cite{GSLW19}]
\label{lem:LCBE}
Given $\alpha_{C_1}$-block encoding $O_{C_1}$ of $C_1$ and $\alpha_{C_2}$-block encoding $O_{C_2}$ of $C_2$, 
an $\alpha_C$-block encoding $O_C$ of $C= (C_1+C_2)$ with $\alpha_C = \alpha_{C_1}+\alpha_{C_2}$
can be constructed with one use of controlled-$O_{C_1}$ and controlled-$O_{C_2}$ each.
\end{lemma}

\begin{lemma}[Optimal scaling quantum linear system algorithm {\cite{CAS+22}}]
\label{lem:opt_lin}
    Let $O_C$ be an $\alpha_C$-block encoding of $C$. Let $O_b$ be the oracle preparing the initial state $\ket{b}$. Then the quantum state
    \begin{equation}
        \frac{C^{-1}\ket{b}}{\norm{C^{-1}\ket{b}}}
    \end{equation}
    can be prepared with accuracy $\epsilon$ using
    \begin{equation}
    \label{eq:opt_lin_cost}
        \mathbf{O}\left(\alpha_C\alpha_{C^{-1}}\log\left(\frac{1}{\epsilon}\right)\right)
    \end{equation}
    queries to controlled-$O_C$, controlled-$O_b$, and their inverses, where $\alpha_{C^{-1}}\geq\norm{C^{-1}}$ is an upper bound on the norm of the inverse matrix.
\end{lemma}
It is worth emphasizing that Lemma~\ref{lem:opt_lin} reports the worst-case complexity of preparing the target state with high ${\bf O}(1)$ probability. We state all other state-preparation results in this paper under the same assumption.

\subsection{Previous results on quantum eigenvalue and phase estimation}
\label{sec:bgpreviouswork}
In this section, we briefly review the existing algorithms for QEVE and QPE, 
and comment on the limitations of their techniques. 
Let us begin the discussion by reviewing the technique of standard (unitary)
phase estimation~\cite{Kit95,NC00} and Hermitian eigenvalue estimation~\cite{KOS07}. In the standard
QPE algorithm, a unitary $U$ and one of its eigenvectors $\ket{\lambda}$ are 
given, and the goal is to estimate $\phi = \arg(\lambda)$ to a precision $\epsilon$.
This is achieved by generating a state of the form
\begin{equation}
\label{eq:QPEstate}
    \frac{1}{\sqrt{2^a}}\sum_j\ket{j}U^j\ket{\lambda} = \frac{1}{\sqrt{2^a}}\sum_je^{i\phi j}\ket{j}\ket{\lambda},
\end{equation}
which is also called the Fourier state in the literature~\cite{LS24}. Next, applying the inverse quantum Fourier transform on the first register yields
\begin{equation}
    \frac{1}{\sqrt{2^a}}\sum_je^{2\pi i\phi j}\ket{j}\ket{\lambda} \mapsto 
    \ket{\tilde{\phi}}\ket{\lambda},
\end{equation}
where $\ket{\tilde{\phi}}$ has high overlap over those basis states 
$\{\ket{j}\}$ for which $j$ is an $a$-bit approximation of $\phi/2\pi$.
The crucial step in the standard QPE algorithm is the preparation of 
the Fourier state in Eq.~\eqref{eq:QPEstate}, which is achieved by successive applications of
powers of controlled-$U$ with control on the ancilla register.
This algorithm for QPE immediately leads to an algorithm for QEVE for Hermitian matrices.
Given a Hermitian matrix $H$, one can use 
Hamiltonian simulation techniques to generate a unitary $U: = e^{-2\pi iH/\norm{H}}$.
Now the QPE algorithm described above can be used as a subroutine to obtain
an approximation of $\lambda/\norm{H}$, where $\lambda\in \mathbb{R}$ is 
the eigenvalue of $H$ corresponding to the given eigenvector $\ket{\lambda}$.

Suppose now that $U$ is non-unitary with eigenvalues on the unit circle
and is given by a block encoding $O_U$. In this case,  
implementing controlled-$U^j$ is exponentially 
expensive in the integer power $j$ in general~\cite{LS24}. Therefore, a straightforward extension 
of the standard QPE leads to exponential complexity in at least one of the 
input parameters. To overcome this limitation, a novel efficient approach to the generation of
a state similar to the one in Eq.~\eqref{eq:QPEstate} was mandated. 

While preparation of the Fourier state in Eq.~\eqref{eq:QPEstate} remains an open
problem for a given non-unitary $U$,\footnote{Although, the potential of the 
Cauchy integral approach~\cite{TOSU22} has not been investigated. Also, the techniques 
in Ref.~\cite{JL23} could be relevant.} 
an efficient approach to prepare this state for a non-Hermitian matrix $A$
was found and led to an efficient algorithm for QEVE.
This approach based on differential equation-solver
requires $A$ to be diagonalizable, i.e. $A=SDS^{-1}$, and have real spectrum.
Both these conditions can be relaxed to some extent, but the algorithm loses efficiency
when $A$ has complex eigenvalues with positive imaginary part.
The Fourier state is prepared by observing that it is the solution history state
for the differential equation $d\ket{\psi}/dt = -iA\ket{\psi}$.
Combined with the state-of-the-art techniques for solving differential equations, 
this approach was reported to achieve an asymptotic query complexity 
${\bf \tilde{O}}(\norm{A}\bar{\kappa}_A^2/\epsilon)$,
where ${\bf \tilde{O}}(f) = {\bf O}(f \text{polylog}(f))$ 
denotes polylogarithmic correction~\cite{Sha22}. However, the analysis in this work 
ignores the errors in the inverse Fourier transform step, and the correct complexity 
of this approach is unknown.
The complexity of this algorithm is also slightly suboptimal 
with respect to its dependence on $\alpha_A/\epsilon$. 

The dependence on $\alpha_A/\epsilon$ was improved later 
by an algorithm based on ``Chebyshev history state'' generation, 
which achieves optimal ${\bf O}(\norm{A}/\epsilon)$ dependence, 
known as the Heisenberg scaling. The technique also leverages
QLSA to generate a normalized state proportional to
\begin{equation}
    \sum_j\ket{j}\tilde{T}_j(A/\alpha_A)\ket{\psi} \propto 
    \sum_l \beta_l \left(\sum_j \cos(2\pi j \phi_l)\ket{j}\right)\ket{\lambda_l},
\end{equation}
where $\{\tilde{T}_j\}$ denote rescaled Chebyshev polynomials of the first kind, and
$\phi_l = \arccos(\lambda_l/\alpha_A)/2\pi$. After generating this step, an application of inverse quantum Fourier
transform yields an approximation of $\phi_l$ in the ancilla register. The condition number of the matrix to be inverted using QLSA scales as $K_{U}(A) = \sum_{j=1}^{n} \norm{U_j(A/\alpha_A)}$, 
where $\{U_j\}$ denote Chebyshev polynomials of the second kind, and $n \in {\bf O}(\alpha_A/\epsilon)$. 

In Ref.~\cite{LS24}, the input state to QEVE is assumed to be an eigenstate $\ket{\lambda_l}$, instead of a superposition as considered in the previous work and our work. The goal of QEVE formulated in 
Ref.~\cite{LS24} is to return a classical estimate of $\lambda_l$.
For this easier problem, the complexity of the approach based on Chebyshev history state generation is ${\bf O}(\alpha_A K_U(A))/\epsilon$. A subtlety here is that the problem formulated in Ref.~\cite{LS24} could be solved more efficiently by other approaches based on expectation value estimation
with complexity ${\bf O}(\alpha_A/\epsilon)$ independent of $K_U(A)$. To be more 
rigorous, we state this in the form of a proposition below.

\begin{proposition}[Quantum eigenvalue estimation by expectation value estimation]
    We are given an accuracy parameters $\epsilon_{\rm eig} \in \mathbb{R}^+$, a lower bound on the
    probability of failure $p_{\rm fail} \in (0,1)$, 
    an $\alpha_A$-block encoding $O_A$ of an $n$-qubit matrix $A$ with real spectrum, 
    and an $n$-qubit unitary circuit $O_{\lambda_l}$ such that $O_\lambda\ket{0}=\ket{\tilde{\lambda}_l}$ prepares an initial state within distance 
    $\norm{\ket{\tilde{\lambda}_l}-\ket{\lambda_l}}\le \epsilon_{\rm eig}/2\alpha_A$ from an eigenstate $\ket{\lambda_l}$ of $A$ 
    with unknown eigenvalue $\lambda_l$. 
    Then, $\lambda_l$ can be estimated with accuracy $\epsilon_{\rm eig}$ 
    and probability $1-p_{\rm fail}$ using
    \begin{equation}
        \mathbf{O}\left(\frac{\alpha_A}{\epsilon_{\rm eig}}\log\left(\frac{1}{p_{\rm fail}}\right)\right)
    \end{equation}
    queries to controlled-$O_A$, controlled-$O_\lambda$, and their inverses.
\end{proposition}
\begin{proof}
    By noting that 
    \begin{equation}
        \lambda_l = \braket{\lambda_l|A|\lambda_l} = \braket{\lambda_l|(A+A^\dagger)/2|\lambda_l},
    \end{equation}
    this task can be achieved using the algorithm for 
    computing the expectation value of a block encoded matrix~\cite{Ral20,ANB+22}.
\end{proof}

Nevertheless, the techniques developed in Ref.~\cite{LS24} could be extended to coherent phase estimation. 
For Problem~\ref{prob:inf_qeve}, the complexity would still scale as  ${\bf O}\left(\alpha_A K_U(A)/\epsilon_{\rm eig}\right)$
or worse for constant $\kappa_S$ and $\epsilon_{\rm st}$.
For diagonalizable $A$, we have $K_U(A) \in {\bf O}(\bar{\kappa}_A)$, so that the total complexity scales as ${\bf O}\left(\alpha_A \bar{\kappa}_A/\epsilon_{\rm eig}\right)$. If $A$ is non-diagonalizable and has the largest Jordan block of size $d$, then $K_U(A) \in {\bf O}(\bar{\kappa}_A \alpha_A^{d-1}/\epsilon^{d-1})$. Note that the total complexity of the approach based on Chebyshev history state scales as ${\bf O}(\alpha_A^d)$ with respect to $\alpha_A$.

The discussion above also suggests why extension of the existing approach for QEVE
to solving QERE would be inefficient. In the Fourier-state based approach, the complexity
of solving the relevant differential equation grows exponentially with total time $T$
when eigenvalues have positive imaginary part~\cite{Kro23}. For eigenvalue estimation, 
$T \propto 1/\epsilon_{\rm eig}$,
so we expect exponential dependence on $1/\epsilon_{\rm eig}$ for this algorithm. 
For Chebyshev history state-based algorithm, the complexity depends on 
the condition number of the matrix to be inverted, which scales as $K_U$. 
Since Chebyshev polynomials $U_j$ grow exponentially with $j$
for complex values, and maximum value of $j$ is proportional to $1/\epsilon$, we
expect the complexity of this approach to also be exponential in $1/\epsilon$.

\section{Estimation of eigenvalues on the unit circle}
\label{sec:qpe}

In this section, we first define the problem of estimation of unimodular eigenvalues rigorously, the informal statement for which is given in Problem~\ref{prob:inf_qpe}. Then we design the algorithm based on resolvent-state generation, 
and finally prove the main result by performing error and complexity analysis. 

First, we provide a rigorous definition for the state $\ket{\phi_{\lambda}}$ 
in Problem~\ref{prob:inf_qpe}. We keep this definition more general so as 
to be applicable for QERE as well.
\begin{definition}[$\epsilon$-additive parametric quantum approximation]
\label{def:aqa}
For an accuracy $\epsilon \in \mathbb{R}^+$, a parameter $\lambda \in \mathbb{C}$, $a \in \mathbb{Z}^+$ qubits defining a Hilbert space $\mathscr H$ with computational basis $\{\ket{j}\}_{j=0}^{2^a-1}$, and a $2^a$-discretization $[\gamma,2^a]$ of a complex curve $\gamma$, we say that a normalized state $\ket{\phi_{\lambda}} \in \mathcal{H}$ is an {\it $\epsilon$-additive quantum approximation ($\epsilon$-APQA) of $\lambda$} with respect to $[\gamma,2^a]$ if
\begin{equation}\label{eq:aqa}
    \sum_{j:|j/2^a-\gamma^{-1}(\lambda)| \le \epsilon} \abs{\braket{j|\phi_{\lambda_l}}}^2 = 1,
\end{equation}
or equivalently, if $\braket{j|\phi_{\lambda}} =0$ for all  $j:|j/2^a-\gamma^{-1}(\lambda)| > \epsilon$. 
\end{definition}
The intuition behind this definition is that for any measurement result $j$ obtained by measuring  $\ket{\phi_{\lambda}}$ in the computational basis, $j/2^a$ is guaranteed to be $\epsilon-$close
to the parameter value $\gamma^{-1}(\lambda)$ corresponding to $\lambda$. 
In other words,  measuring $\ket{\phi_{\lambda}}$ produces a good estimate of $\lambda$. 
We can also restate Eq.~\eqref{eq:aqa}
in terms of the projector on to the desired subspace and its orthogonal complement such that
\begin{align}\label{eq:proj}
    P_{\lambda,\epsilon} &:=  \sum_{j: |j/2^a-\gamma^{-1}(\lambda)| \le \epsilon} \ket{j}\bra{j} \implies  P_{\lambda,\epsilon}\ket{\phi_{\lambda}} = \ket{\phi_{\lambda}} ,\nonumber\\
Q_{\lambda,\epsilon} &:= \mathds{1}_a -  P_{\lambda,\epsilon} \implies Q_{\lambda,\epsilon}\ket{\phi_{\lambda}} = \bm{0}.
\end{align}
If $\gamma$ is a closed curve, then the limits of the summation in Eq.~\eqref{eq:aqa} in Def.~\ref{def:aqa}
are replaced by 
$j:|j/2^a-\gamma^{-1}(\lambda) \mod 1| \le \epsilon$, to account for the fact that $t=0$ and $t=1$ denote the
same points on the curve.

To provide the formal definition of the problem solved by our algorithm, recall from Problem~\ref{prob:inf_qpe} that we are given a non-unitary matrix $A$ which has some eigenvalues of unit magnitude. We denote these eigenvalues by $\{\lambda_l\}$ and other eigenvalues of $A$ by $\{\mu_k\}$. It is assumed that all eigenvalues in $\{\mu_k\}$ are such that $\abs{\mu_k} \notin (1,1+\epsilon_{\rm eig})$. No assumption is made on the defectiveness (or size of the Jordan block) of the eigenvalues $\{\lambda_l\}$ or $\{\mu_k\}$. We are given an input state $\ket{\psi} = \sum_l \beta_l \ket{\lambda_l}$, where $\ket{\lambda_l}$ is a (rank-$1$) eigenvector of $A$ with eigenvalue $\lambda_l$. The goal is to generate a state close to $\ket{\psi_A}= \sum_{l}\beta_l\ket{\phi_{\lambda_l}}\ket{\lambda_l}$, where $\ket{\phi_{\lambda_l}}$ is an $(\epsilon_{\rm eig}/2\pi)$-APQA of  $\lambda_l$ (see Def.~\ref{def:aqa}). 
Note that $(\epsilon_{\rm eig}/2\pi)$-APQA of $\lambda_l$ ensures that a 
good estimate of $\arg(\lambda_l)$ has at most $\epsilon_{\rm eig}$ additive error. 
We are now ready to state the problem rigorously.

\begin{problem}[Quantum coherent Estimation of unimodular Eigenvalues (QEUE)]
\label{prob:qpe}
We are given two accuracy parameters $\epsilon_{\rm st}, \epsilon_{\rm eig} \in \mathbb{R}^+$, an $\alpha_A$-block encoding $O_A$ of an $n$-qubit matrix $A$ such that $\Lambda(A)\cap \{z:1 < |z| < 1+\epsilon_{\rm eig}\}=\emptyset$, and an $n$-qubit unitary circuit $O_{\psi}$ such that $\ket{\psi}:=O_{\psi}\ket{0^n}$ has support over only the eigenstates $\{\ket{\lambda_l}\}$ of $A$ with eigenvalues $\{\lambda_l\}$ on the unit circle $S^1$, i.e.  
\begin{equation}
\label{eq:qpefinalstate}
    \ket{\psi} = \sum_{l:\lambda_l \in \Lambda(A) \cap S^1} \beta_l \ket{\lambda_l}, \quad \beta_l \in \mathbb{C} \ \forall l.
\end{equation}
The goal is to generate to an accuracy $\epsilon_{\rm st}$ an $(n+a)$-qubit state of the form
\begin{equation}
    \ket{\psi_A} = \frac{\sum_l \beta_l \ket{\phi_{\lambda_l}}\ket{\lambda_l}}
    {\norm{\sum_l \beta_l \ket{\phi_{\lambda_l}}\ket{\lambda_l}}},
\end{equation}
where $\ket{\phi_{\lambda_l}} \in \mathcal{H}_2^{\otimes a}$ for each $l$ is an 
$(\epsilon_{\rm eig}/2\pi)$-APQA of $\lambda_l$ with respect to a discretization of the unit circle $[\gamma,2^a]$. Here
$\gamma:[0,1]\mapsto S^1$ is the curve $\gamma(t)=e^{2\pi it}$.
\end{problem}

The key tool used in our algorithm is the resolvent of $A$, given by $R(z) = (z\mathds{1}_n - A)^{-1}$,  which we use to generate the state $\ket{\psi_A}$ to some desired accuracy. We consider a block diagonal matrix $R$ whose $j$-th block is $R(z_j)$ for a complex number $z_j= (1+\delta)\gamma(t_j)$  $\forall j$, where $\delta\in(0,\epsilon_{\rm eig})$ 
is a parameter whose value we set later. The matrix $R$ can be expressed compactly in terms of a diagonal matrix with entries $\{z_j\}$ as defined below. 
\begin{definition}[Matrix encoding of a discretized curve]
For a $\delta>0$, $a\in\mathbb{Z}^+$, and a discretization $\Gamma=[\gamma, 2^a]$, define
\begin{equation}
    Z_{\delta,\Gamma} := \sum_{j=0}^{2^a-1} z_j \ket{j}\bra{j}, \quad 
    z_j := (1+\delta)\gamma(t_j).
\end{equation}
\end{definition}

Since none of the eigenvalues of $A$ have magnitude $1+\delta$, the resolvent matrix $R:=(Z_{\delta,\Gamma}\otimes \mathds{1}_n - \mathds{1}_a \otimes A)^{-1}$ is bounded. Our algorithm proceeds by generating an approximation to the state $\ket{\tilde{\psi}_{A}} \propto R\ket{+^a}\ket{\psi}$, which we call the resolvent state, using QLSA. Towards this step, we first generate the block encoding of 
\begin{equation}
\label{eq:A}
M: = R^{-1} = Z_{\delta,\Gamma}\otimes \mathds{1}_n - \mathds{1}_a \otimes A
\end{equation}
in the following lemma.
\begin{lemma}[Block encoding of $M$]
Given the inputs of Problem~\ref{prob:qpe}, a $\delta>0$, and a discretization $\Gamma=[\gamma, 2^a]$, then an $\alpha_M$-block encoding $O_M$ of $M$ with $\alpha_M = 1+\delta +\alpha_A $ can be constructed with one use of $O_A$.
\end{lemma}
\begin{proof}
Since $\norm{Z_{\delta,\Gamma}} = 1+\delta$, we can construct a $(1+\delta)$-block encoding of $Z_{\delta,\Gamma}$. Then the desired 
$\alpha_M$-block encoding of $M$ can be constructed by linear combination using Lemma~\ref{lem:LCBE}.
\end{proof}

The complexity of generating the state $\ket{\tilde{\psi}_{A}}$ using QLSA depends on $\norm{M^{-1}}$. The following lemma establishes a bound on $\norm{M^{-1}}$.
\begin{lemma}[Spectral norm bound on the resolvent]
\label{lem:condA_qpe}
Let $A$ be a complex matrix of size $2^n \times 2^n$ such that $\mathcal{K}^\circ_{\delta}(A) < \infty$. Then $M$ defined in Eq.~\ref{eq:A}  is invertible and
\begin{equation}
    \norm{M^{-1}} \le \frac{\mathcal{K}^\circ_{\delta}(A)}{\delta}.
\end{equation}
\end{lemma}
\begin{proof}
Note that $M = \sum_j \ket{j}\bra{j}\otimes [z_j\mathds{1}_n - A]$, so that $M^{-1} = \sum_j \ket{j}\bra{j}\otimes [z_j\mathds{1}_n - A]^{-1}$ is block diagonal. Consequently, 
\begin{equation}
        \norm{M^{-1}} = \max_j \norm{[z_j\mathds{1}_n - A]^{-1}} \le \sup_{z:|z|=1+\delta}\norm{[z\mathds{1}_n - A]^{-1}} = \mathcal{K}^\circ_\delta(A)/\delta,
    \end{equation}
    following the definition of $\mathcal{K}^\circ_\delta(A)$ in Def.~\ref{def:Kdeltacirc}.
\end{proof}

We are now ready to derive the complexity of generating the state proportional to $R\ket{+^a}\ket{\psi}$.
\begin{lemma}[Complexity of preparing the resolvent state]
\label{lem:stateprep_qpe}
We are given the inputs of Problem~\ref{prob:qpe}, a $\delta\in(0,\epsilon_{\rm eig})$, an $a\in\mathbb{Z}^+$, and $\mathcal{K}^\circ_{\delta}(A)$. Let
    \begin{equation}\label{eq:psiUtilde}
        \ket{\tilde{\psi}_{A}}:=\sum_l \beta_l\ket{\tilde{\phi}_{\lambda_l,\delta}}\ket{\lambda_l},\quad
      \ket{\tilde{\phi}_{\lambda_l,\delta}}  := 
      \sqrt{\frac{\delta(2+\delta)}{2^a}}\sum_{j=0}^{2^a-1} (z_j-\lambda_l)^{-1}\ket{j} \quad \forall l.
    \end{equation}
Then the state $\ket{\tilde{\psi}_{A}}/\norm{\ket{\tilde{\psi}_{A}}}$ can be prepared with accuracy $\epsilon_{\rm st}/2$ using
    \begin{equation}
        \mathbf{O}\left(\frac{\alpha_A\mathcal{K}_{\delta}(A)}{\delta}
        \log\left(\frac{1}{\epsilon_{\rm st}}\right)\right)
    \end{equation}
queries to controlled-$O_A$,  controlled-$O_{\psi}$ and their inverses.
\end{lemma}
\begin{remark}
\label{rem:prefactor}
The above lemma holds even if the prefactor  $\sqrt{\delta(2+\delta)/2^a}$  does not accompany the definition of  $\ket{\tilde{\phi}_{\lambda_l,\delta}}$.  In general, the norm of $\ket{\tilde{\phi}_{\lambda_l,\delta}}$ depends on $\lambda_l$. However, for large $a$ and small $\delta$, the norm is nearly independent of $\lambda_l$ as we will show later in Lemma~\ref{lem:qpe_errors}, and the factor $\sqrt{\delta(2+\delta)/2^a}$ approximately normalizes $\ket{\tilde{\phi}_{\lambda_l,\delta}}$.
\end{remark}
\begin{proof}
Observe that
\begin{equation}
    \frac{\ket{\tilde{\psi}_{A}}}{\norm{\ket{\tilde{\psi}_{A}}}} = \frac{M^{-1}\ket{+^a}\ket{\psi}}{\norm{M^{-1}\ket{+^a}\ket{\psi}}}.
\end{equation}
The complexity follows from the complexity of QLSA in Lemma~\ref{lem:opt_lin}.
\end{proof}

The resolvent state $\ket{\tilde{\psi}_{A}}/\norm{\ket{\tilde{\psi}_{A}}}$ is our candidate for the approximation of the desired state $\ket{\psi_A}$. We now analyze the errors in order to determine suitable values of $\delta$ and $a$ to meet the requirements on $\ket{\psi_A}$ stated in Problem~\ref{prob:qpe}. The value of $\delta$ in turn will also dictate the complexity of our algorithm due to Lemma~\ref{lem:stateprep_qpe}.

\begin{lemma}[Upper bound on discretization error]
For $\delta \le 1$ and 
\begin{equation}
	2^a \ge \frac{(3\sqrt{2}\pi+6)\epsilon_{\rm eig}}{\delta^3},
\end{equation}
projectors $ P_{\lambda_l,\epsilon}, Q_{\lambda_l,\epsilon}$ defined in Eq.~\eqref{eq:proj} 
with $\epsilon:=\epsilon_{\rm eig}/2\pi$,
and state $\ket{\tilde{\phi}_{\lambda_l,\delta}}$ defined in Eq.~\eqref{eq:psiUtilde}, we have
\begin{align}
\abs{\norm{\ket{\tilde{\phi}_{\lambda_l,\delta}}}^2-
\int_{0}^{1} g_l(t) dt} &\le \frac{\delta}{\epsilon_{\rm eig}}, \nonumber\\
\abs{\norm{P_{\lambda_l,\epsilon}\ket{\tilde{\phi}_{\lambda_l,\delta}}}^2-
\int_{t_-}^{t_+} g_l(t) dt} &\le \frac{\delta}{\epsilon_{\rm eig}},
\end{align}
where 
\begin{equation}
        g_l(t) := \frac{\delta(2+\delta)}{\abs{(1+\delta)e^{2\pi it}-\lambda_l}^2}=\frac{\delta(2+\delta)}{\delta^2+4(1+\delta)\sin^2(\pi t-\arg(\lambda_l)/2)}
\end{equation}
and
$t_\pm = (\arg(\lambda_l) \pm \epsilon_{\rm eig})/2\pi$.
\end{lemma}

\begin{proof}
In this proof, we treat $t$ and $j$ as modular variables, i.e. 
$t = t \mod 1$ and $j = j \mod 2^a$. 
Observe that 
\begin{align}
\label{eq:normtosum}
\norm{\ket{\tilde{\phi}_{\lambda_l,\delta}}}^2 &= 
\frac{\delta(2+\delta)}{2^a}\sum_{j=0}^{2^a-1}\frac{1}{\abs{z_j-\lambda_l}^2}\\
\norm{P_{\lambda_l,\epsilon}\ket{\tilde{\phi}_{\lambda_l,\delta}}}^2 &=
\frac{\delta(2+\delta)}{2^a}\sum_{j=\lceil Nt_-\rceil }^{\lfloor Nt_+\rfloor}\frac{1}{\abs{z_j-\lambda_l}^2}.
\end{align}
Using the cosine formula, we get
\begin{equation}
    \abs{(1+\delta)e^{2\pi it}-\lambda_l}^2 = 1+(1+\delta)^2-2(1+\delta)\cos(2\pi t-\arg(\lambda_l)) = \delta^2+4(1+\delta)\sin^2(\pi t_j-\arg(\lambda_l)/2).
\end{equation}
Now we can view $\norm{\ket{\tilde{\phi}_{\lambda_l,\delta}}}^2$ and  
$\norm{P_{\lambda_l,\epsilon}\ket{\tilde{\phi}_{\lambda_l,\delta}}}^2$
as discrete integrals of $g_l$. More precisely,
\begin{equation}
\norm{\ket{\tilde{\phi}_{\lambda_l,\delta}}}^2 = 
\sum_{0}^{1}[g_l,2^a],\quad 
\norm{P_{\lambda_l,\epsilon}\ket{\tilde{\phi}_{\lambda_l,\delta}}}^2 =
\sum_{t_-}^{t_+}[g_l,2^a].
\end{equation}
In the next step, we will approximate these discrete integrals by their continuous counterparts.
To bound $d g_l/dt$, we first define $x = 2\sqrt{1+\delta}\sin((\pi t-\arg(\lambda_l)/2))$,
so that $g_l(t) = \delta(2+\delta)/(\delta^2+x^2)$. Then
\begin{align}
\abs{\frac{d g_l}{d t}} &= \abs{\frac{2\delta x}{(\delta^2+x^2)^2} \left((2+\delta) \frac{dx}{dt}\right)} \\
&\le \abs{\frac{2\delta x}{(\delta^2+x^2)}\frac{1}{(\delta^2+x^2)}\left((2+\delta) 2\pi\sqrt{1+\delta}\right)} \\
&\le \frac{1}{\delta^2}\left(2\pi(2+\delta) \sqrt{1+\delta}\right) \\
&\le \frac{6\sqrt{2}\pi}{\delta^2}.
\end{align}
Observe that 
\begin{equation}
\abs{g_l(t)} = \abs{\frac{\delta(2+\delta)}{(\delta^2+x^2)}}\le \abs{\frac{\delta(2+\delta)}{\delta^2}}
\le 1+\frac{2}{\delta}.
\end{equation}
Now using Lemma~\ref{lem:riemannsum}, we have
\begin{equation}
    \label{eq:discerror}
        \abs{\int_{0}^{1} g_l(t)dt - \sum_{0}^{1}[g_l,2^a]} 
        \le \epsilon_{\rm disc}, \quad
        \abs{\int_{t_-}^{t_+} g_l(t)dt - \sum_{t_-}^{t_+}[g_l,2^a]} \le 
        \epsilon_{\rm disc},
\end{equation}
where
\begin{equation}
\epsilon_{\rm disc} = \frac{1}{2^a}\left\{ \frac{6\sqrt{2}\pi}{2\delta^2} + \frac{4}{\delta} +2\right\} 
\le \frac{3\sqrt{2}\pi+6}{2^a \delta^2}. 
\end{equation}
Finally,
\begin{equation}
2^a \ge \frac{(3\sqrt{2}\pi+6)\epsilon_{\rm eig}}{\delta^3} \implies 
\epsilon_{\rm disc} \le \frac{\delta}{\epsilon_{\rm eig}}.
\end{equation}

\end{proof}

\begin{lemma}[Failure probability of eigenvalue estimation]
\label{lem:qpe_errors}
For $\delta \le \epsilon_{\rm eig}/4$ and 
\begin{equation}
	2^a \ge \frac{(3\sqrt{2}\pi+6)\epsilon_{\rm eig}}{\delta^3},
\end{equation}
projectors $ P_{\lambda_l,\epsilon}, Q_{\lambda_l,\epsilon}$ defined in Eq.~\eqref{eq:proj} and state $\ket{\tilde{\phi}_{\lambda_l,\delta}}$ defined in Eq.~\eqref{eq:psiUtilde}, the following inclusions hold for each $l$:
    \begin{enumerate}
        \item \begin{equation}
            \frac{\sqrt{5}}{2} \ge \norm{P_{\lambda_l,\epsilon}\ket{\tilde{\phi}_{\lambda_l,\delta}}} 
      \ge \frac{1}{2}.
        \end{equation}
        \item \begin{equation}
        \abs{1-\frac{\norm{P_{\lambda_0,\epsilon}\ket{\tilde{\phi}_{\lambda_0,\delta}}}}{\norm{P_{\lambda_l,\epsilon}\ket{\tilde{\phi}_{\lambda_l,\delta}}}}} 
        \le \frac{4\delta}{\epsilon_{\rm eig}}.
    \end{equation}
    \item \begin{equation}
        \norm{Q_{\lambda_l,\epsilon}\ket{\tilde{\phi}_{\lambda_l,\delta}}}       
        \le \sqrt{2(1+1/\pi)}\sqrt{\frac{\delta}{\epsilon_{\rm eig}}}.
    \end{equation}
    \end{enumerate}

\end{lemma}

\begin{proof}
Let $a_l = \norm{P_{\lambda_l}\ket{\tilde{\phi}_{\lambda_l,\delta}}}$
and $b_l = \norm{Q_{\lambda_l}\ket{\tilde{\phi}_{\lambda_l,\delta}}}$. 
\begin{enumerate}
\item We have
\begin{equation}
a_l^2 = \norm{P_{\lambda_l,\epsilon}\ket{\tilde{\phi}_{\lambda_l,\delta}}}^2 =
\sum_{t_-}^{t_+}[g_l,2^a]
\end{equation}
and
\begin{equation}
        \abs{\int_{t_-}^{t_+} g_l(t)dt - \sum_{t_-}^{t_+}[g_l,2^a]}  = 
        \abs{\int_{t_-}^{t_+} g_l(t)dt - a_l^2} \le 
        \epsilon_{\rm disc} \le \frac{\delta}{\epsilon_{\rm eig}}\le 1/4.
\end{equation}
By explicit integration, we get
\begin{equation}
\int_{t_-}^{t_+} g_l(t)dt = \frac{2}{\pi}\arctan{\eta}, \quad
\eta = \frac{2+\delta}{\delta}\tan(\epsilon_{\rm eig}/2).
\end{equation}
Since $\epsilon_{\rm eig} \ge \delta$, we have 
\begin{equation}
\eta = \frac{2+\delta}{\delta}\tan(\epsilon_{\rm eig}/2) \ge 
\frac{2+\delta}{\delta}\left(\frac{\epsilon_{\rm eig}}{2}\right) \ge 1
\end{equation}
and therefore $\int_{t_-}^{t_+} g_l(t)dt \ge 1/2$.
Consequently, we have 
\begin{equation}
a_l^2 = \norm{P_{\lambda_l,\epsilon}\ket{\tilde{\phi}_{\lambda_l,\delta}}}^2 \ge 
\int_{t_-}^{t_+} g_l(t)dt - \epsilon_{\rm disc} \ge 1/4.
\end{equation}
Consequently $a_l\ge 1/2$ for all $l$.
Similarly, 
\begin{equation}
a_l^2 =  \le 
\int_{t_-}^{t_+} g_l(t)dt + \epsilon_{\rm disc} \le 1 + \frac{\delta}{\epsilon_{\rm eig}} \le \frac{5}{4}.
\end{equation}

\item We have $\abs{a_l^2-a_0^2} \le 2\epsilon_{\rm disc}$. Then using $a_l \ge 1/2$, we get
\begin{equation}
\abs{1-\frac{a_0}{a_l}} = \abs{\frac{a_l^2-a_0^2}{a_l(a_0+a_l)}} \le \frac{2\epsilon_{\rm disc}}{1/2}
\le \frac{4\delta}{\epsilon_{\rm eig}}.
\end{equation}

\item Finally,
\begin{equation}
b_l^2 = \norm{Q_{\lambda_l,\epsilon}\ket{\tilde{\phi}_{\lambda_l,\delta}}}^2 =
\left(\sum_{0}^{1} - \sum_{t_-}^{t_+}\right)[g_l,2^a].
\end{equation}
We then have 
\begin{equation}
        \abs{\left(\int_{0}^{1} - \int_{t_-}^{t_+}\right) g_l(t)dt - 
        \left(\sum_{0}^{1} - \sum_{t_-}^{t_+}\right)[g_l,2^a]}  \le 
        2\epsilon_{\rm disc} \le 2\delta/\epsilon_{\rm eig}.
\end{equation}
The integral can be bounded as follows:
\begin{equation}
\left(\int_{0}^{1} - \int_{t_-}^{t_+}\right) g_l(t)dt = 1 - \frac{2}{\pi}\arctan{\eta} = 
\frac{2}{\pi}\text{arccot}({\eta}).
\end{equation}
Let $\theta = \text{arccot}({\eta})$, or equivalently, $\cot(\theta) = \eta$.
Then 
\begin{equation}
\theta \le \tan(\theta) = \frac{\delta}{2+\delta}\cot(\epsilon_{\rm eig}) \le 
\frac{\delta}{2+\delta}\left(\frac{2}{\epsilon_{\rm eig}}\right)
\le \frac{\delta}{\epsilon_{\rm eig}}.
\end{equation}
Consequently, we have 
\begin{equation}
\left(\int_{0}^{1} - \int_{t_-}^{t_+}\right) g_l(t)dt \le \frac{2\delta}{\pi \epsilon_{\rm eig}}.
\end{equation}
Finally,
\begin{equation}
b_l^2 \le \frac{2\delta}{\pi \epsilon_{\rm eig}}+2\epsilon_{\rm disc} \implies b_l \le 
\sqrt{(2+2/\pi)\frac{\delta}{\epsilon_{\rm eig}}}.
\end{equation}
\end{enumerate}
\end{proof}

Our final theorem requires bounding the error in approximating $\ket{\psi_A}$ in terms of the errors in  individual $\ket{\tilde{\phi}_{\lambda_l,\delta}}$. The following lemma allows us to accomplish this.
\begin{lemma}[Propagation of failure in eigenvalue estimation]
    Let $\ket{\psi} = \sum_l \beta_l \ket{\lambda_l}$ be such that $\norm{\ket{\psi}} = 1$. Let 
    $S$ be a rectangular matrix with columns $\{\ket{\lambda_l}\}$.
    Then for any set of ancilla states $\{\ket{\chi_l} \in \mathcal{H}'\}$,
    \begin{equation}
       \frac{\min_l\norm{\ket{\chi_l}}}{\kappa_S}  \le \norm{\sum_l \beta_l \ket{\chi_l}\ket{\lambda_l}} \le \kappa_S \max_l\norm{\ket{\chi_l}},
    \end{equation}    
    where $\kappa_S$ denotes the condition number of $S$.
\end{lemma}
\begin{proof}
    Since the columns of $S$ are $\ket{\lambda_l}$, we have
    $\ket{\lambda_l} = S\ket{e_l}$, where $\{\ket{e_l}\}$ denote the
    computational basis states.
    First note that $\ket{\psi} =S\sum_l \beta_l \ket{e_l}$, which together with 
    $\norm{\ket{\psi}}=1$ yields
\begin{equation}
    \sigma_{\max}(S)\norm{\sum_l \beta_l \ket{e_l}} \ge 1 ,\quad 
    \sigma_{\min}(S)\norm{\sum_l \beta_l \ket{e_l}} \le 1.
\end{equation}
We therefore get
\begin{equation}
\label{eq:normbeta}
        \frac{1}{\sigma_{\max}(S)} \le \sqrt{\sum_l \abs{\beta_l}^2} \le \frac{1}{\sigma_{\min}(S)}.
\end{equation}
Now
\begin{align}
    \norm{\sum_l \beta_l \ket{\chi_l}\ket{\lambda_l}} &= \norm{S\sum_l \beta_l \ket{\chi_l}\ket{e_l}}\\
    &\le \sigma_{\max}(S)\norm{\sum_l \beta_l \ket{\chi_l}\ket{e_l}}\\
    &\le \sigma_{\max}(S)\sqrt{\sum_l \abs{\beta_l}^2 \norm{\ket{\chi_l}}^2}\\
    &\le \sigma_{\max}(S)\max_l\norm{\ket{\chi_l}}\sqrt{\sum_l \abs{\beta_l}^2 }\\
    &\le \kappa_S\max_l\norm{\ket{\chi_l}},
\end{align}
where in the last step we used Eq.~\eqref{eq:normbeta}. The other inequality can be proved similarly with the starting point
\begin{equation}
    \norm{\sum_l \beta_l \ket{\chi_l}\ket{\lambda_l}} \ge 
    \sigma_{\min}(S)\norm{\sum_l \beta_l \ket{\chi_l}\ket{e_l}}.
\end{equation}
\end{proof}

\begin{lemma}
\label{lem:normalizederror}
Let $\ket{x}$ and $\ket{\tilde{x}}$ be such that $\norm{\ket{x}-\ket{\tilde{x}}} \le \epsilon$ for some 
$\epsilon \in \mathbb{R}^+$. Then 
\begin{equation}
\norm{\frac{\ket{x}}{\norm{\ket{x}}} - \frac{\ket{\tilde{x}}}{\norm{\ket{\tilde{x}}}}} \le \frac{2\epsilon}{\norm{\ket{x}}}
\end{equation}
\end{lemma}

\begin{proof}
\begin{align}
\norm{\frac{\ket{x}}{\norm{\ket{x}}} - \frac{\ket{\tilde{x}}}{\norm{\ket{\tilde{x}}}}} &\le
\norm{\frac{\ket{x}}{\norm{\ket{x}}} - \frac{\ket{\tilde{x}}}{\norm{\ket{x}}}} + \norm{\frac{\ket{\tilde{x}}}{\norm{\ket{x}}} - \frac{\ket{\tilde{x}}}{\norm{\ket{\tilde{x}}}}} \nonumber\\
&\le \frac{\epsilon}{\norm{\ket{x}}} + \norm{\ket{\tilde{x}}}\frac{\abs{\norm{\ket{x}} - \norm{\ket{\tilde{x}}}}}{\norm{\ket{x}}\norm{\ket{\tilde{x}}}} \nonumber \\
& \le  \frac{\epsilon}{\norm{\ket{x}}} + \frac{\epsilon}{\norm{\ket{x}}} = \frac{2\epsilon}{\norm{\ket{x}}},
\end{align}
where we used $\abs{\norm{\ket{x}} - \norm{\ket{\tilde{x}}}} \le \norm{\ket{x}-\ket{\tilde{x}}} \le \epsilon$ 
by triangle inequality
in the penultimate step.
\end{proof}

We are now ready to state and prove the main theorem of this section.
This theorem stipulates the values of the parameters $\delta, a$ 
in our algorithm that ensure that $\ket{\tilde{\psi}_A}$ is $\epsilon_{\rm st}$-close
to $\ket{\psi_A}$ as required.
\begin{thm}[Complexity of solving QEUE]
\label{thm:qpeideal}
Given the inputs of Problem~\ref{prob:qpe},  a quantum algorithm can be constructed that generates $\ket{\psi_A}$ to accuracy $\epsilon_{\rm st}$ by making
    \begin{equation}
        \mathbf{O}\left(\frac{ \alpha_A \kappa_S^4\mathcal{K}^\circ_\delta(A)}{\epsilon_{\rm eig}\epsilon_{\rm st}^2} 
        \log\left(\frac{1}{\epsilon_{\rm st}}\right)\right)
    \end{equation}
queries to controlled-$O_A$, controlled-$O_{\psi}$ and their inverses, where
$\delta \in \Theta(\epsilon_{\rm eig}\epsilon_{\rm st}^2/\kappa_S^4)$.
\end{thm}
\begin{proof}
Observe that $\ket{\tilde{\psi}_A}= \sum_{l} \beta_l \ket{\tilde{\phi}_{\lambda_l,\delta}}\ket{\lambda_l}$
can be expressed as
\begin{multline}
        \ket{\tilde{\psi}_A} = 
         \underbrace{\sum_{l} \beta_l \left(\frac{\norm{P_{\lambda_0,\epsilon}\ket{\tilde{\phi}_{\lambda_0,\delta}}}}{\norm{P_{\lambda_l,\epsilon}\ket{\tilde{\phi}_{\lambda_l,\delta}}}}\right)P_{\lambda_l,\epsilon}\ket{\tilde{\phi}_{\lambda_l,\delta}}\ket{\lambda_l}}_{\ket{\tilde{\psi}_{A,1}}} \\+ 
        \underbrace{\sum_{l} \beta_l \left(1-\frac{\norm{P_{\lambda_0,\epsilon}\ket{\tilde{\phi}_{\lambda_0,\delta}}}}{\norm{P_{\lambda_l,\epsilon}\ket{\tilde{\phi}_{\lambda_l,\delta}}}}\right)P_{\lambda_l,\epsilon}\ket{\tilde{\phi}_{\lambda_l,\delta}}\ket{\lambda_l}}_{\ket{\tilde{\psi}_{A,2}}}\\
        +\underbrace{\sum_{l} \beta_l Q_{\lambda_l,\epsilon}\ket{\tilde{\phi}_{\lambda_l,\delta}}\ket{\lambda_l}}_{\ket{\tilde{\psi}_{A,3}}}.
    \end{multline}
Now using Lemma~\ref{lem:qpe_errors}, 
\begin{align}
        \norm{\ket{\tilde{\psi}_{A,1}}} &\ge \norm{P_{\lambda_0,\epsilon}\ket{\tilde{\phi}_{\lambda_l,\delta}}}/\kappa_S 
        = \frac{a_0}{\kappa_S} \ge \frac{1}{2\kappa_S}\\
        \norm{\ket{\tilde{\psi}_{A,2}}} &\le \kappa_S\max_l \left(1-\frac{\norm{P_{\lambda_0,\epsilon}\ket{\tilde{\phi}_{\lambda_0,\delta}}}}{\norm{P_{\lambda_l,\epsilon}\ket{\tilde{\phi}_{\lambda_l,\delta}}}}\right)\norm{P_{\lambda_l,\epsilon}\ket{\tilde{\phi}_{\lambda_l,\delta}}} \le \kappa_S\frac{2\sqrt{5}\delta}{\epsilon_{\rm eig}}\\
        \norm{\ket{\tilde{\psi}_{A,3}}} &\le \kappa_S\max_l \norm{Q_{\lambda_l,\epsilon}\ket{\tilde{\phi}_{\lambda_l,\delta}}}
        \le \kappa_S\sqrt{2(1+1/\pi)\frac{\delta}{\epsilon_{\rm eig}}}.
    \end{align}
    We therefore have
    \begin{equation}
        \frac{\norm{\ket{\tilde{\psi}_{A,2}}}}{\norm{\ket{\tilde{\psi}_{A,1}}}}
        \le \frac{4\sqrt{5}\delta\kappa_S^2}{\epsilon_{\rm eig}}, \quad 
        \frac{\norm{\ket{\tilde{\psi}_{A,3}}}}{\norm{\ket{\tilde{\psi}_{A,1}}}}
        \le \kappa_S^2\sqrt{8(1+1/\pi)\frac{\delta}{\epsilon_{\rm eig}}}
    \end{equation}
For 
\begin{equation}
\delta \le \min\left\{\frac{\epsilon_{\rm st}\epsilon_{\rm eig}}{32\sqrt{5}\kappa_S^2}, \frac{\epsilon_{\rm st}^2\epsilon_{\rm eig}}{512(1+1/\pi)\kappa_S^4}\right\}, 
\end{equation}
we get
 \begin{equation}
        \frac{\norm{\ket{\tilde{\psi}_{A,2}}}}{\norm{\ket{\tilde{\psi}_{A,1}}}}
        \le \frac{\epsilon_{\rm st}}{8}, \quad 
        \frac{\norm{\ket{\tilde{\psi}_{A,3}}}}{\norm{\ket{\tilde{\psi}_{A,1}}}}
        \le \frac{\epsilon_{\rm st}}{8}.    
\end{equation}
We now set $\ket{x} = \ket{\tilde{\psi}_{A,1}}$ and $\ket{\tilde{x}} = \ket{\tilde{\psi}_{A}}$. 
Then we have 
\begin{equation}
\frac{\norm{\ket{\tilde{x}} - \ket{x}}}{\norm{\ket{x}}} \le 
\frac{\norm{\ket{\tilde{\psi}_{A,2}}}}{\norm{\ket{\tilde{\psi}_{A,1}}}} +
\frac{\norm{\ket{\tilde{\psi}_{A,3}}}}{\norm{\ket{\tilde{\psi}_{A,1}}}} \le \frac{\epsilon_{\rm st}}{4}.
\end{equation}
Equivalently, we have 
\begin{equation}
\frac{\norm{\ket{\tilde{x}} - \ket{x}}}{\norm{\ket{x}}} \le \frac{\epsilon_{\rm st}\norm{\ket{x}}}{4}.
\end{equation}
Then by Lemma~\ref{lem:normalizederror}, we have 
\begin{equation}
\norm{\frac{\ket{x}}{\norm{\ket{x}}} - \frac{\ket{\tilde{x}}}{\norm{\ket{\tilde{x}}}}} \le \frac{2\epsilon_{\rm st}}{4} 
= \frac{\epsilon_{\rm st}}{2}.
\end{equation}
If QLSA returns the state $\ket{\tilde{x}}/\norm{\ket{\tilde{x}}} = \ket{\tilde{\psi}_A}/\norm{\ket{\tilde{\psi}_A}}$ with accuracy $\epsilon_{\rm st}/2$ as in Lemma \ref{lem:stateprep_qpe}, then the total error will be less than $\epsilon_{\rm st}$. 

The algorithm for preparing the state $\ket{\tilde{\psi}_A}/\norm{\ket{\tilde{\psi}_A}}$ and its query complexity follows from 
Lemma~\ref{lem:stateprep_qpe} and the fact that 
\begin{equation}
\frac{1}{\delta} \in {\bf O}\left(\frac{\kappa_S^4}{\epsilon_{\rm st}^2\epsilon_{\rm eig}}\right).
\end{equation}
\end{proof}

\begin{remark}
In Theorem~\ref{thm:qpeideal}, $\mathcal{K}^\circ_\delta(A)$ can be replaced by
any $K$ such that $\mathcal{K}^\circ_\delta(A) \le K$ 
that is given instead of $\mathcal{K}^\circ_\delta(A)$. In particular, 
if the dimension of the largest Jordan block $d$ and 
Jordan condition number $\bar{\kappa}_A$ are given, then by 
Proposition~\ref{prop:kreissbound}, $\mathcal{K}^\circ_\delta(A)$
can be replaced by $K = \bar{\kappa}_A(1-\delta^d)/\delta^{d-1}(1-\delta)$, 
resulting in the total query complexity
\begin{equation}
    {\bf O}\left(\frac{\alpha_A\bar{\kappa}_A\kappa_S^{4d}}{\epsilon_{\rm eig}^d \epsilon_{\rm st}^{2d}}
    \log\left(\frac{1}{\epsilon_{\rm st}}\right)\right) \in
    {\bf O}\left(\frac{\alpha_A\bar{\kappa}_A^{4d+1}}{\epsilon_{\rm eig}^d \epsilon_{\rm st}^{2d}}
    \log\left(\frac{1}{\epsilon_{\rm st}}\right)\right).
\end{equation}
For diagonalizable $A$, we have $d=1$, which leads to a bound on the total query complexity 
\begin{equation}
    {\bf O}\left(\frac{\alpha_A\bar{\kappa}_A\kappa_S^4}{\epsilon_{\rm eig} \epsilon_{\rm st}^2}
    \log\left(\frac{1}{\epsilon_{\rm st}}\right)\right) \in 
    {\bf O}\left(\frac{\alpha_A\bar{\kappa}_A^{5}}{\epsilon_{\rm eig}^d \epsilon_{\rm st}^{2d}}
    \log\left(\frac{1}{\epsilon_{\rm st}}\right)\right).
\end{equation}
\end{remark}
\begin{remark}
    Notice that we choose $a \in \text{polylog}(\alpha_A, \kappa_S, 1/\epsilon_{\rm eig}, 1/\epsilon_{\rm st})$ in our algorithm, and therefore the number of ancilla qubits and additional two-qubit gates required is larger than the total query complexity by a factor at most $\text{polylog}(\alpha_A, \kappa_S, 1/\epsilon_{\rm eig}, 1/\epsilon_{\rm st})$. 
\end{remark}

\section{Estimation of real eigenvalues}
\label{sec:qeve}

In this section, we first define the problem of eigenvalue estimation rigorously, the informal statement for which is given in Problem~\ref{prob:inf_qeve}. Then we design the algorithm based on resolvent-state generation, and finally prove the main result by performing error and complexity analysis. The 
problem statement as well as the analysis of the algorithm is very similar to that of 
QEUE in \S~\ref{sec:qpe}.

We now provide the formal definition of the problem solved by our algorithm. Recall from Problem~\ref{prob:inf_qeve} that we are given a matrix $A$ whose spectrum is not necessarily real. We denote the real eigenvalues of $A$ by $\{\lambda_l\}$, and other eigenvalues by $\{\mu_k\}$. It is assumed that $\Im(\mu_k) \notin (0,\epsilon_{\rm eig})$ for all $k$. No assumption is made on the defectiveness of the eigenvalues $\{\lambda_l\}$ and $\{\mu_k\}$. We are given a state $\ket{\psi} = \sum_l \beta_l \ket{\lambda_l}$, where $\ket{\lambda_l}$ is a eigenvector of $A$ with eigenvalue $\lambda_l$. In other words,  $\ket{\lambda_l}$ is a generalized eigenvector with rank one. The goal is to generate a state close to $\ket{\psi_A}= \sum_{l}\beta_l\ket{\phi_{\lambda_l}}\ket{\lambda_l}$, where $\ket{\phi_{\lambda_l}}$ is an $\epsilon_{\rm eig}$-APQA of  $\lambda_l$, 
as defined in Def.~\ref{def:aqa}. We are now ready to state the problem rigorously.

\begin{problem}[Quantum coherent Estimation of Real Eigenvalues (QERE)]
\label{prob:qeve}
We are given two accuracy parameters $\epsilon_{\rm st}, \epsilon_{\rm eig} \in \mathbb{R}^+$, an $\alpha_A$-block encoding $O_A$ of an $n$-qubit matrix $A$ such that $\Im(\Lambda(A))\cap(0,\epsilon_{\rm eig})=\emptyset$, and an $n$-qubit unitary circuit $O_{\psi}$ such that $\ket{\psi}:=O_{\psi}\ket{0^n}$ has support over only the eigenstates $\{\ket{\lambda_l}\}$ of $A$ with real eigenvalues $\{\lambda_l\}$, i.e.  
\begin{equation}
    \ket{\psi} = \sum_{l:\lambda_l \in \Lambda(A) \cap \mathbb{R}} \beta_l \ket{\lambda_l}, \quad \beta_l \in \mathbb{C} \ \forall l.
\end{equation}
The goal is to generate to an accuracy $\epsilon_{\rm st}$ an $(n+a)$-qubit state of the form
\begin{equation}
    \ket{\psi_A} = \frac{\sum_l \beta_l \ket{\phi_{\lambda_l}}\ket{\lambda_l}}
    {\norm{\sum_l \beta_l \ket{\phi_{\lambda_l}}\ket{\lambda_l}}},
\end{equation}
where $\ket{\phi_{\lambda_l}} \in \mathcal{H}_2^{\otimes a}$ for each $l$ is an $(\epsilon_{\rm eig}/\rho)$-APQA of $\lambda_l$ with respect to a discrete curve $[\gamma,2^a]$. Here $\gamma$ is the curve 
\begin{equation}
    \gamma:[0,1]\to [-\rho,\rho],\quad  \gamma(t) = \rho(2t-1),\quad  \rho=\alpha_A + \epsilon_{\rm eig}.
\end{equation}
\end{problem}
Even though the eigenvalues $\{\lambda_l\}$ lie on the 
segment $[-\alpha_A,\alpha_A]$, we choose $\text{Image}(\gamma)$ to be 
$[-\alpha_A-\epsilon_{\rm eig},\alpha_A+\epsilon_{\rm eig}]$
in our problem. This is to ensure that an approximation with precision $\epsilon_{\rm eig}$ of any eigenvalue $\lambda_l$ indeed lies on $\gamma$. Similarly, the problem requires preparation of 
$(\epsilon_{\rm eig}/2\rho)$-AQPA of $\lambda_l$, which ensures that a good estimate has
at most $\epsilon_{\rm eig}$ additive error.

The design of the algorithm reflects that of the algorithm for QEUE. We start by constructing
the resolvent matrix $R$ with diagonal block $R(z_j) = (z\mathds{1}_n - A)^{-1}$. However, 
in this case, we choose $z_j= \gamma(t_j) + i\delta$ for some $0<\delta<\epsilon_{\rm eig}$.
The matrix $R$ can be expressed compactly in terms of a diagonal matrix with entries $\{z_j\}$ 
as defined below. 
\begin{definition}
For a $0<\delta<\epsilon_{\rm eig}$, $a\in\mathbb{Z}^+$, and a discretization $\Gamma=[\gamma, 2^a]$, define
\begin{equation}
    Z_{\delta,\Gamma} := \sum_{j=0}^{2^a-1} z_j \ket{j}\bra{j}, \quad 
    z_j := \gamma(t_j) + i\delta.
\end{equation}
\end{definition}

Since none of the eigenvalues of $A$ have imaginary part $\delta$, the resolvent matrix $R:=(Z_{\delta,\Gamma}\otimes \mathds{1}_n - \mathds{1}_a \otimes A)^{-1}$ is bounded. Our algorithm proceeds by generating an approximation to the state $\ket{\tilde{\psi}_{A}} \propto R\ket{+^n}\ket{\psi}$ using QLSA. This state, in turn, provides an approximation of the desired state $\ket{\psi_A}$. Towards the former step, we first generate the block encoding of 
\begin{equation}
\label{eq:A_H}
M: = R^{-1} = Z_{\delta,\Gamma}\otimes \mathds{1}_n - \mathds{1}_a \otimes A
\end{equation}
in the following lemma.
\begin{lemma}[Block encoding of $M$]
Given the inputs of Problem~\ref{prob:qeve}, a $\delta>0$, and a discretization $\Gamma=[\gamma, 2^a]$, then an $\alpha_M$-block encoding $O_M$ of $M$ with $\alpha_M = \sqrt{\rho^2+\delta^2} +\alpha_A $ can be constructed with one use of $O_A$.
\end{lemma}
\begin{proof}
Since $\norm{Z_{\delta,\Gamma}} = \sqrt{\rho^2+\delta^2}$, we can construct a $(\sqrt{\rho^2+\delta^2})$-block encoding of $Z_{\delta,\Gamma}$. Then the desired 
$\alpha_M$-block encoding of $M$ can be constructed by linear combination using Lemma~\ref{lem:LCBE}.
\end{proof}

The complexity of generating the state $\ket{\tilde{\psi}_{A}}$ using QLSA depends on $\norm{M^{-1}}$. The following lemma establishes bounds on $\norm{M^{-1}}$.
\begin{lemma}[Spectral norm bound on the resolvent]
\label{lem:condA}
Let $A$ be a complex matrix of size $2^n \times 2^n$ such that $\mathcal{K}_{\delta}(-iA) < \infty$ for some $\delta>0$. Then $M$ defined in Eq.~\eqref{eq:A_H}  is invertible and
\begin{equation}
    \norm{M^{-1}} \le \frac{\mathcal{K}_{\delta}(-iA)}{\delta}.
\end{equation}
\end{lemma}
\begin{proof}
Note that $M = \sum_j \ket{j}\bra{j}\otimes [z_j\mathds{1}_n - A]$, so that $M^{-1} = \sum_j \ket{j}\bra{j}\otimes [(\gamma(t_j)+i\delta)\mathds{1}_n - A]^{-1}$ is block diagonal. Consequently, 
    \begin{equation}
        \norm{M^{-1}} = \max_j \norm{[(\gamma(t_j)+i\delta)\mathds{1}_n - A]^{-1}} \le \sup_{x\in\mathbb{R}}\norm{[(\delta-ix)\mathds{1}_n + iA]^{-1}} =
        \frac{\mathcal{K}_\delta(-iA)}{\delta},
    \end{equation}
    following the definition of $\mathcal{K}_\delta$ in Def.~\ref{def:Kdelta}.
\end{proof}

We are now ready to derive the complexity of generating the state proportional to $R\ket{+^a}\ket{\psi}$.
\begin{lemma}[Complexity of generating the resolvent state]
\label{lem:stateprepqere}
We are given the inputs of Problem~\ref{prob:qeve}, a $\delta\in(0,\epsilon_{\rm eig})$, an $a\in\mathbb{Z}^+$, and $\mathcal{K}_{\delta}(-iA)$. Let
    \begin{equation}\label{eq:psiHtilde}
        \ket{\tilde{\psi}_{A}}:=\sum_l \beta_l\ket{\tilde{\phi}_{\lambda_l,\delta}}\ket{\lambda_l},\quad
      \ket{\tilde{\phi}_{\lambda_l,\delta}}  := 
      \sqrt{\frac{(\alpha_A+\epsilon_{\rm eig})\delta}{2^{a-1}\pi}}\sum_{j=0}^{2^a-1} (z_j-\lambda_l)^{-1}\ket{j} \quad \forall l.
    \end{equation}
Then the state $\ket{\tilde{\psi}_{A}}/\norm{\ket{\tilde{\psi}_{A}}}$ can be prepared with accuracy $\epsilon_{\rm st}/2$ using
    \begin{equation}
        \mathbf{O}\left(\frac{\alpha_A\mathcal{K}_{\delta}(-iA)}{\delta}
        \log\left(\frac{1}{\epsilon_{\rm st}}\right)\right)
    \end{equation}
queries to controlled-$O_A$,  controlled-$O_{\psi}$ and their inverses.
\end{lemma}
\begin{remark}
Similar to Remark~\ref{rem:prefactor}, we have added an $l$-independent prefactor  $\sqrt{\rho\delta/2^{a-1}\pi}$, for $\rho = \alpha_A + \epsilon_{\rm eig}$, to the definition of $\ket{\tilde{\phi}_{\lambda_l,\delta}}$ in Eq.~\eqref{eq:psiHtilde}. This prefactor, even if absent, does not affect the above lemma because we explicitly normalize $\ket{\tilde{\psi}_{A}}$ in the lemma's proof. 
However, we note that  for large $a$ and small $\delta$ the added prefactor approximately normalizes $\ket{\tilde{\phi}_{\lambda_l,\delta}}$ and will be useful in simplifying future calculations. We prove this fact in  Lemma~\ref{lem:qere_errors}.

\end{remark}
\begin{proof}
Observe that
\begin{equation}
    \frac{\ket{\tilde{\psi}_{A}}}{\norm{\ket{\tilde{\psi}_{A}}}} = \frac{M^{-1}\ket{+^a}\ket{\psi}}{\norm{M^{-1}\ket{+^a}\ket{\psi}}}.
\end{equation}
The complexity follows from the complexity of QLSA in Lemma~\ref{lem:opt_lin}.
\end{proof}

The state $\ket{\tilde{\psi}_{A}}/\norm{\ket{\tilde{\psi}_{A}}}$
is our candidate for the approximation of the desired state $\ket{\psi_A}$. We now analyze the errors in order to determine suitable values of $\delta$ and $a$ to meet the requirements on $\ket{\psi_A}$ stated in Problem~\ref{prob:qeve}. The value of $\delta$ in turn will also dictate the complexity of our algorithm due to Lemma~\ref{lem:stateprepqere}.

\begin{lemma}[Upper bound on discretization error]
For $\delta \le 1$ and 
\begin{equation}
	2^a \ge \frac{5\rho \epsilon_{\rm eig}}{\pi\delta^3},
\end{equation}
projectors $ P_{\lambda_l,\epsilon}, Q_{\lambda_l,\epsilon}$ 
defined in Eq.~\eqref{eq:proj} 
with $\epsilon:=\epsilon_{\rm eig}/2\rho$ 
and state $\ket{\tilde{\phi}_{\lambda_l,\delta}}$ defined in Eq.~\eqref{eq:psiUtilde}, we have
\begin{align}
\abs{\norm{\ket{\tilde{\phi}_{\lambda_l,\delta}}}^2-
\int_{0}^{1} f_l(t) dt} &\le \frac{\delta}{\epsilon_{\rm eig}}, \nonumber\\
\abs{\norm{P_{\lambda_l,\epsilon}\ket{\tilde{\phi}_{\lambda_l,\delta}}}^2-
\int_{t_-}^{t_+} f_l(t) dt} &\le \frac{\delta}{\epsilon_{\rm eig}},
\end{align}
where 
\begin{equation}
        f_l(t) := \frac{2\rho\delta}{\pi[(\gamma(t)-\lambda_l)^2+\delta^2]} = 
        \frac{2\rho\delta}{\pi[(\rho (2t-1)-\lambda_l)^2+\delta^2]} 
\end{equation}
and
$t_\pm = \gamma^{-1}(\lambda_l) \pm \epsilon = 
\gamma^{-1}(\lambda_l \pm \epsilon_{\rm eig})$, $\rho = \alpha_A+\epsilon_{\rm eig}$.
\end{lemma}

\begin{proof}
Observe that 
\begin{align}
\norm{\ket{\tilde{\phi}_{\lambda_l,\delta}}}^2 &= 
\frac{2\rho\delta}{\pi 2^a}\sum_{j=0}^{2^a-1}\frac{1}{\abs{z_j-\lambda_l}^2}\\
\norm{P_{\lambda_l,\epsilon}\ket{\tilde{\phi}_{\lambda_l,\delta}}}^2 &=
\frac{2\rho\delta}{\pi 2^a}\sum_{j=\lceil Nt_-\rceil }^{\lfloor Nt_+\rfloor}\frac{1}{\abs{z_j-\lambda_l}^2}.
\end{align}
Using Pythagorus theorem, we get
\begin{equation}
    \abs{z(t)-\lambda_l}^2 = (\rho (2t-1)-\lambda_l)^2+\delta^2.
\end{equation}
Now we can view $\norm{\ket{\tilde{\phi}_{\lambda_l,\delta}}}^2$ and  
$\norm{P_{\lambda_l,\epsilon}\ket{\tilde{\phi}_{\lambda_l,\delta}}}^2$
as discrete integrals of $g_l$. More precisely,
\begin{equation}
\norm{\ket{\tilde{\phi}_{\lambda_l,\delta}}}^2 = 
\sum_{0}^{1}[f_l,2^a],\quad 
\norm{P_{\lambda_l,\epsilon}\ket{\tilde{\phi}_{\lambda_l,\delta}}}^2 =
\sum_{t_-}^{t_+}[f_l,2^a].
\end{equation}
In the next step, we will approximate these discrete integrals by their continuous counterparts.
To bound $d f_l/dt$, we first define $x = \rho (2t-1)-\lambda_l$,
so that $f_l(t) = 2\rho\delta/\pi(\delta^2+x^2)$. Then
\begin{align}
\abs{\frac{d f_l}{d t}} &= \abs{\frac{4\rho\delta x}{\pi(\delta^2+x^2)^2}}  \\
&\le \abs{\frac{2\delta x}{(\delta^2+x^2)}\frac{2\rho}{\pi(\delta^2+x^2)}} \\
&\le \frac{2\rho}{\pi\delta^2}.
\end{align}
Observe that 
\begin{equation}
\abs{f_l(t)} = \abs{\frac{2\rho\delta}{\pi(\delta^2+x^2)}}\le \frac{2\rho\delta}{\pi\delta^2}
=\frac{2\rho}{\pi\delta}.
\end{equation}
Now using Lemma~\ref{lem:riemannsum}, we have
\begin{equation}
        \abs{\int_{0}^{1} f_l(t)dt - \sum_{0}^{1}[f_l,2^a]} 
        \le \epsilon_{\rm disc}, \quad
        \abs{\int_{t_-}^{t_+} f_l(t)dt - \sum_{t_-}^{t_+}[f_l,2^a]} \le 
        \epsilon_{\rm disc},
\end{equation}
where
\begin{equation}
\epsilon_{\rm disc} = \frac{1}{2^a}\left\{\frac{2\rho}{2\pi\delta^2}+2\frac{2\rho}{\pi\delta}\right\} = 
\frac{1}{2^a}\left(\frac{\rho}{\pi\delta^2}+\frac{4\rho}{\pi\delta}\right)
\le \frac{1}{2^a}\left(\frac{\rho}{\pi\delta^2}+\frac{4\rho}{\pi\delta^2}\right)
\le \frac{1}{2^a}\frac{5\rho}{\pi\delta^2}
\end{equation}
Finally,
\begin{equation}
2^a \ge \frac{5\rho \epsilon_{\rm eig}}{\pi\delta^3} \implies 
\epsilon_{\rm disc} \le \frac{\delta}{\epsilon_{\rm eig}}.
\end{equation}

\end{proof}

\begin{lemma}[Failure probability of eigenvalue estimation]
\label{lem:qere_errors}
For $\delta \le \epsilon_{\rm eig}/4$ and 
\begin{equation}
	2^a \ge \frac{5\rho \epsilon_{\rm eig}}{\pi\delta^3},
\end{equation}
projectors $ P_{\lambda_l,\epsilon}, Q_{\lambda_l,\epsilon}$ defined in Eq.~\eqref{eq:proj} and state $\ket{\tilde{\phi}_{\lambda_l,\delta}}$ defined in Eq.~\eqref{eq:psiUtilde}, the following inclusions hold for each $l$:
    \begin{enumerate}
        \item \begin{equation}
            \frac{\sqrt{5}}{2} \ge \norm{P_{\lambda_l,\epsilon}\ket{\tilde{\phi}_{\lambda_l,\delta}}} 
      \ge \frac{1}{2}.
        \end{equation}
        \item \begin{equation}
        \abs{1-\frac{\norm{P_{\lambda_0,\epsilon}\ket{\tilde{\phi}_{\lambda_0,\delta}}}}{\norm{P_{\lambda_l,\epsilon}\ket{\tilde{\phi}_{\lambda_l,\delta}}}}} 
        \le \frac{4\delta}{\epsilon_{\rm eig}}.
    \end{equation}
    \item \begin{equation}
        \norm{Q_{\lambda_l,\epsilon}\ket{\tilde{\phi}_{\lambda_l,\delta}}}       
        \le \sqrt{2(1+1/\pi)}\sqrt{\frac{\delta}{\epsilon_{\rm eig}}}.
    \end{equation}
    \end{enumerate}

\end{lemma}

\begin{proof}
Let $a_l = \norm{P_{\lambda_l}\ket{\tilde{\phi}_{\lambda_l,\delta}}}$
and $b_l = \norm{Q_{\lambda_l}\ket{\tilde{\phi}_{\lambda_l,\delta}}}$. 
\begin{enumerate}
\item We have
\begin{equation}
a_l^2 = \norm{P_{\lambda_l,\epsilon}\ket{\tilde{\phi}_{\lambda_l,\delta}}}^2 =
\sum_{t_-}^{t_+}[f_l,2^a]
\end{equation}
and
\begin{equation}
        \abs{\int_{t_-}^{t_+} f_l(t)dt - \sum_{t_-}^{t_+}[f_l,2^a]}  = 
        \abs{\int_{t_-}^{t_+} f_l(t)dt - a_l^2} \le 
        \epsilon_{\rm disc} \le \frac{\delta}{\epsilon_{\rm eig}}\le 1/4.
\end{equation}
By explicit integration, we get
\begin{equation}
\int_{t_-}^{t_+} f_l(t)dt = 
\frac{2}{\pi} \arctan \left(\frac{\epsilon_{\rm eig}}{\delta}\right)
\end{equation}
Since $\epsilon_{\rm eig} \ge \delta$, we have 
\begin{equation}
\frac{2}{\pi} \arctan \left(\frac{\epsilon_{\rm eig}}{\delta}\right) \ge \frac{2}{\pi}\left(\frac{\pi}{4}\right) \ge 1/2
\end{equation}
and therefore $\int_{t_-}^{t_+} g_l(t)dt \ge 1/2$.
Consequently, we have 
\begin{equation}
a_l^2 = \norm{P_{\lambda_l,\epsilon}\ket{\tilde{\phi}_{\lambda_l,\delta}}}^2 \ge 
\int_{t_-}^{t_+} g_l(t)dt - \epsilon_{\rm disc} \ge 1/4.
\end{equation}
Consequently $a_l\ge 1/2$ for all $l$.
Similarly, 
\begin{equation}
a_l^2 =  
\int_{t_-}^{t_+} g_l(t)dt + \epsilon_{\rm disc} \le 1 + \frac{\delta}{\epsilon_{\rm eig}} \le \frac{5}{4}.
\end{equation}

\item We have $\abs{a_l^2-a_0^2} \le 2\epsilon_{\rm disc}$. Then using $a_l \ge 1/2$, we get
\begin{equation}
\abs{1-\frac{a_0}{a_l}} = \abs{\frac{a_l^2-a_0^2}{a_l(a_0+a_l)}} \le \frac{2\epsilon_{\rm disc}}{1/2}
\le \frac{4\delta}{\epsilon_{\rm eig}}.
\end{equation}

\item Finally,
\begin{equation}
b_l^2 = \norm{Q_{\lambda_l,\epsilon}\ket{\tilde{\phi}_{\lambda_l,\delta}}}^2 =
\left(\sum_{0}^{1} - \sum_{t_-}^{t_+}\right)[f_l,2^a].
\end{equation}
We then have 
\begin{equation}
        \abs{\left(\int_{0}^{1} - \int_{t_-}^{t_+}\right) f_l(t)dt - 
        \left(\sum_{0}^{1} - \sum_{t_-}^{t_+}\right)[f_l,2^a]}  \le 
        2\epsilon_{\rm disc} \le 2\delta/\epsilon_{\rm eig}.
\end{equation}
The integral can be bounded as follows:
\begin{equation}
\left(\int_{0}^{1} - \int_{t_-}^{t_+}\right) f_l(t)dt\le 
\left(\int_{-\infty}^{\infty} - \int_{t_-}^{t_+}\right) f_l(t)dt
 = 1 - \frac{2}{\pi}\arctan{\frac{\epsilon_{\rm eig}}{\delta}} = 
\frac{2}{\pi}\text{arccot}(\frac{\epsilon_{\rm eig}}{\delta}).
\end{equation}
Let $\theta = \text{arccot}(\frac{\epsilon_{\rm eig}}{\delta})$, or equivalently, 
$\cot(\theta) = \frac{\epsilon_{\rm eig}}{\delta}$.
Then 
\begin{equation}
\theta \le \tan(\theta) =  \frac{\delta}{\epsilon_{\rm eig}}.
\end{equation}
Consequently, we have 
\begin{equation}
\left(\int_{0}^{1} - \int_{t_-}^{t_+}\right) g_l(t)dt \le \frac{2\delta}{\pi \epsilon_{\rm eig}}.
\end{equation}
Finally,
\begin{equation}
b_l^2 \le \frac{2\delta}{\pi \epsilon_{\rm eig}}+2\epsilon_{\rm disc} \implies b_l \le 
\sqrt{(2+2/\pi)\frac{\delta}{\epsilon_{\rm eig}}}.
\end{equation}
\end{enumerate}
\end{proof}

We are now ready to state and prove the main theorem of this section.
This theorem stipulates the values of the parameters $\delta, a$ 
in our algorithm that ensure that $\ket{\tilde{\psi}_A}$ is $\epsilon_{\rm st}$-close
to $\ket{\psi_A}$ as required.
\begin{thm}[Complexity of solving QERE]
\label{thm:qeveideal}
Given the inputs of Problem~\ref{prob:qeve},  a quantum algorithm can be constructed that generates $\ket{\psi_A}$ to accuracy $\epsilon_{\rm st}$ by making
    \begin{equation}
        \mathbf{O}\left(\frac{ \alpha_A \kappa_S^4\mathcal{K}_\delta(-iA)}{\epsilon_{\rm eig}\epsilon_{\rm st}^2} 
         \log\left(\frac{1}{\epsilon_{\rm st}}\right)\right)
    \end{equation}
queries to controlled-$O_A$, controlled-$O_{\psi}$ and their inverses, where
$\delta \in \Theta(\epsilon_{\rm eig}\epsilon_{\rm st}^2/\kappa_S^4)$.
\end{thm}
\begin{proof}
The proof is identical to that of Theorem~\ref{thm:qpeideal}, except that Lemma~\ref{lem:stateprepqere}
is used in place of Lemma~\ref{lem:stateprep_qpe} in the last step.
\end{proof}

\begin{remark}
In Theorem~\ref{thm:qeveideal}, $\mathcal{K}_\delta(-iA)$ can be replaced by
any $K$ such that $\mathcal{K}_\delta(-iA) \le K$ 
that is given instead of $\mathcal{K}_\delta(-iA)$. In particular, 
if the dimension of the largest Jordan block $d$ and 
Jordan condition number $\bar{\kappa}_A$ are given, then by 
Proposition~\ref{prop:kreissbound}, $\mathcal{K}_\delta(-iA)$
can be replaced by $K = \bar{\kappa}_A(1-\delta^d)/\delta^{d-1}(1-\delta)$, 
resulting in the total query complexity
\begin{equation}
    {\bf O}\left(\frac{\alpha_A\bar{\kappa}_A\kappa_S^{4d}}{\epsilon_{\rm eig}^d \epsilon_{\rm st}^{2d}}
    \log\left(\frac{1}{\epsilon_{\rm st}}\right)\right) \in 
    {\bf O}\left(\frac{\alpha_A\bar{\kappa}_A^{4d+1}}{\epsilon_{\rm eig}^d \epsilon_{\rm st}^{2d}}
    \log\left(\frac{1}{\epsilon_{\rm st}}\right)\right).
\end{equation}
For diagonalizable $A$, we have $d=1$, which leads to the total query complexity 
\begin{equation}
    {\bf O}\left(\frac{\alpha_A\bar{\kappa}_A\kappa_S^4}{\epsilon_{\rm eig} \epsilon_{\rm st}^2}
    \log\left(\frac{1}{\epsilon_{\rm st}}\right)\right) \in 
    {\bf O}\left(\frac{\alpha_A\bar{\kappa}_A^5}{\epsilon_{\rm eig} \epsilon_{\rm st}^2}
    \log\left(\frac{1}{\epsilon_{\rm st}}\right)\right).
\end{equation}
\end{remark}
\begin{remark}
    Notice that we choose $a \in \text{polylog}(\alpha_A, \kappa_S, 1/\epsilon_{\rm eig}, 1/\epsilon_{\rm st})$ in our algorithm, and therefore the number of ancilla qubits and additional two-qubit gates required is larger than the total query complexity by a factor at most $\text{polylog}(\alpha_A, \kappa_S, 1/\epsilon_{\rm eig}, 1/\epsilon_{\rm st})$. 
\end{remark}

\section{Estimating eigenvalues on a general curve}
\label{sec:peve}

In this section, we derive the conditions on a general complex curve $\gamma$ under which our algorithm can be extended to the estimation of eigenvalues lying on $\gamma$. The curve $\gamma$ is assumed to be fixed, i.e. not an input to the problem and therefore is problem-size independent.

Let us first analyze our approach for estimating real and unimodular eigenvalues respectively on a common footing, to motivate the extension to  general curves. Observe that in both these problems, the eigenvalues are promised to lie on a curve $\gamma$, but the resolvent matrix is constructed for a discrete curve $[\gamma_\delta, 2^a]$. In the case of phase estimation, the resolvent matrix is constructed using $\{z_j=(1+\delta)\gamma(t_j)\}$ lying on the curve $\gamma_\delta=(1+\delta)\gamma$. Similarly, in the case of real eigenvalue estimation, $\text{Image}(\gamma) = [-\rho,\rho]$, and the resolvent matrix is constructed using $\{z_j=\gamma(t_j)+i\delta\}$ lying on the shifted curve $\gamma_\delta=\gamma+i\delta$. The assumption on the spectrum ensures that no eigenvalues lie on $\gamma_\delta$ as well as within distance $\delta$ on either side $\gamma_\delta$. This ensures that the resolvent state can be generated in complexity linear in ${\bf O}(K/\delta)$, where $K$ is a constant capturing the non-Hermiticity of the input matrix. 

The complexity then depends on how large a $\delta$ can be chosen so that we can prepare the desired state $\ket{\psi_H}$ to required accuracy $\epsilon_{\rm st}$. Notice that while $\gamma$ is fixed, we have the freedom in choosing the curve $\gamma_\delta$ parametrized by $\delta$. The properties of $\gamma$ enter in the error analysis in Lemmas~\ref{lem:qpe_errors} and~\ref{lem:qere_errors}:
    \begin{enumerate}
        \item To obtain bounds on 
\begin{equation*}
            \norm{P_{\lambda_l,\epsilon}\ket{\tilde{\phi}_{\lambda_l,\delta}}}^2 \quad
            \text{and}\quad
        \abs{1-\frac{\norm{P_{\lambda_0,\epsilon}\ket{\tilde{\phi}_{\lambda_0,\delta}}}}{\norm{P_{\lambda_l,\epsilon}\ket{\tilde{\phi}_{\lambda_l,\delta}}}}},
        \end{equation*}
we needed to prove that for every $\epsilon_{\rm eig}>0$,
\begin{equation}
\label{eq:cond1}
     \int_{t-\epsilon}^{t+\epsilon} {\abs{\gamma(t)-\gamma_{\delta}(t')}^{-2}dt'} = c_{\delta,\epsilon} \quad \forall t \in [0,1], \quad \forall \epsilon > 0,
\end{equation}
where $c_{\delta,\epsilon}$ is some constant independent of $t$, and
we also showed that 
\begin{equation}
\label{eq:cond2}
    \abs{1-\frac{\int_{t-\epsilon}^{t+\epsilon} {\abs{\gamma(t)-\gamma_{\delta}(t')}^{-2}dt'}}{\int_{0}^{1} {\abs{\gamma(t)-\gamma_{\delta}(t')}^{-2}dt'}}} \in {\bf O}\left(\frac{\delta}{\epsilon}\right) \quad \forall \epsilon > 0.
\end{equation}

\item Finally to account for discretization errors, we proved that
\begin{equation}
\label{eq:cond3}
    \sup_{t'}\abs{\frac{d}{dt'}\left(\abs{\gamma(t)-\gamma_{\delta}(t')}^{-2}\right)} \in 
    \text{poly}\left(\frac{1}{\delta}\right).
\end{equation}
\end{enumerate}
If these three properties are satisfied for the curve $\gamma$ and the family of curves $\{\gamma_{\delta}\}$, then an efficient algorithm for eigenvalue estimation can be constructed following our approach. We state these requirements formally in the next conjecture.

\begin{conjecture}
Suppose $\gamma$ is a fixed, known complex curve, which admits a family of curves $\{\gamma_\delta,\  \delta \in \mathbb{R}^+\}$ satisfying Equations~\eqref{eq:cond1}, \eqref{eq:cond2} and \eqref{eq:cond3}. We are given two accuracy parameters $\epsilon_{\rm st}, \epsilon_{\rm eig} \in \mathbb{R}^+$, an $\alpha_A$-block encoding $O_A$ of an $n$-qubit diagonalizable matrix $A$ such that a part of $\Lambda(A)$ lies on $\gamma$ and the other part lies outside the $\epsilon_{\rm eig}$- neighborhood of $\gamma$. 
Suppose $\bar{\kappa}_A$ is the Jordan condition number of $A$ and the dimension of 
the largest Jordan block of $A$ is $d \in {\bf O}(\log(n))$.
We are also given an $n$-qubit unitary circuit $O_{\psi}$ such that $\ket{\psi}:=O_{\psi}\ket{0^n}$ has support over only the eigenstates $\{\ket{\lambda_l}\}$ of $A$ with eigenvalues $\{\lambda_l\}$ on $\gamma$, i.e.  
\begin{equation}
\label{eq:qerefinalstate}
    \ket{\psi} = \sum_{l:\lambda_l \in \text{Image}(\gamma)} \beta_l \ket{\lambda_l}, \quad \beta_l \in \mathbb{C} \ \forall l.
\end{equation}
Then a quantum algorithm can generate, with 
$\text{poly}(\alpha_M,\bar{\kappa}_A,1/\epsilon_{\rm eig},1/\epsilon_{\rm st})$ queries to $O_A$ and $O_{\psi}$,
and to an accuracy $\epsilon_{\rm st}$, an $(n+a)$-qubit state of the form
\begin{equation}
    \ket{\psi_A} = \frac{\sum_l \beta_l \ket{\phi_{\lambda_l}}\ket{\lambda_l}}
    {\norm{\sum_l \beta_l \ket{\phi_{\lambda_l}}\ket{\lambda_l}}},
\end{equation}
where $\ket{\phi_{\lambda_l}} \in \mathcal{H}_2^{\otimes a}$ for each $l$ is an $\epsilon_{\rm eig}$-APQA of $\lambda_l$ with respect to a discretization $[\gamma,2^a]$. 
\end{conjecture}

It is then interesting to obtain a simpler characterization of the curves $\gamma$ that satisfy Equations~\eqref{eq:cond1},\eqref{eq:cond2} and \eqref{eq:cond3}. It is clear from the outset that intersecting curves violate Eq.~\eqref{eq:cond1}. Similarly, Eq.~\eqref{eq:cond3} may be violated at discontinuities in a curve. It is an interesting open question whether a non-intersecting and piecewise continuous curve $\gamma$ always admits a family of curves $\{\gamma_\delta\}$ 
satisfying the required properties.  

One possible way of extending our framework could be by applying the algorithm for QEUE on circles centered at the origin with radii $k\delta$, $k = 1, 2, \dots$. Thus, one may use binary search to get the desired eigenvalue even without knowing where this eigenvalue is located, as long as the given eigenvalue is known to not lie on
any of these circles. In this way, prior knowledge of the curve is no longer needed.

\paragraph{Acknowledgments}
This research was supported by the 
Australian Research Council Centre of Excellence for 
Engineered Quantum Systems (CE170100009). 
SK acknowledges support from Deborah Jin Fellowship.
We are grateful for insightful conversations with Kunal Sharma
and Gopikrishnan Muraleedharan.

\bibliographystyle{quantum}
\bibliography{Ref}

\begin{thebibliography}{10}

\bibitem{Kit95}
Alexei~Yu Kitaev.
\newblock ``{Quantum measurements and the Abelian Stabilizer Problem}''~(1995).
\newblock
  \href{http://arxiv.org/abs/quant-ph/9511026}{arXiv:quant-ph/9511026}.

\bibitem{NC00}
Michael~A. Nielsen and Isaac~L. Chuang.
\newblock ``{Quantum Computation and Quantum Information: 10th Anniversary
  Edition}''.
\newblock \href{https://dx.doi.org/10.1017/CBO9780511976667}{Cambridge
  University Press}. Cambridge~(2010).

\bibitem{Sho94}
Peter~W Shor.
\newblock ``Algorithms for quantum computation: {D}iscrete logarithms and
  factoring''.
\newblock In {Proceedings 35th Annual Symposium on Foundations of Computer
  Science}.
\newblock \href{https://dx.doi.org/10.1090/conm/305/05215}{Pages 124--134}.
\newblock IEEE~(1994).

\bibitem{BHMT02}
Gilles Brassard, Peter Hoyer, Michele Mosca, and Alain Tapp.
\newblock ``Quantum amplitude amplification and estimation''.
\newblock \href{https://dx.doi.org/10.1090/conm/305/05215}{Contemp. Math. {\bf
  305}, 53--74}~(2002).

\bibitem{HHL09}
Aram~W. Harrow, Avinatan Hassidim, and Seth Lloyd.
\newblock ``Quantum algorithm for linear systems of equations''.
\newblock \href{https://dx.doi.org/10.1103/PhysRevLett.103.150502}{Phys. Rev.
  Lett. {\bf 103}, 150502}~(2009).

\bibitem{NWZ09}
Daniel Nagaj, Pawel Wocjan, and Yong Zhang.
\newblock ``{Fast amplification of QMA}''.
\newblock \href{https://dx.doi.org/10.5555/2012098.2012106}{Quant. Inf.
  Comput., {\bf 9}, 1053--1068}~(2009).

\bibitem{Ral20}
Patrick Rall.
\newblock ``Quantum algorithms for estimating physical quantities using block
  encodings''.
\newblock \href{https://dx.doi.org/10.1103/PhysRevA.102.022408}{Phys. Rev. A
  {\bf 102}, 022408}~(2020).

\bibitem{DL23}
Zhiyan Ding and Lin Lin.
\newblock ``Even shorter quantum circuit for phase estimation on early
  fault-tolerant quantum computers with applications to ground-state energy
  estimation''.
\newblock \href{https://dx.doi.org/10.1103/PRXQuantum.4.020331}{PRX Quantum
  {\bf 4}, 020331}~(2023).

\bibitem{WBL12}
Nathan Wiebe, Daniel Braun, and Seth Lloyd.
\newblock ``Quantum algorithm for data fitting''.
\newblock \href{https://dx.doi.org/10.1103/PhysRevLett.109.050505}{Phys. Rev.
  Lett. {\bf 109}, 050505}~(2012).

\bibitem{LMR14}
Seth Lloyd, Masoud Mohseni, and Patrick Rebentrost.
\newblock ``Quantum principal component analysis''.
\newblock \href{https://dx.doi.org/10.1038/nphys3029}{Nat. Phys. {\bf 10},
  631--633}~(2014).

\bibitem{KOS07}
Emanuel Knill, Gerardo Ortiz, and Rolando~D. Somma.
\newblock ``Optimal quantum measurements of expectation values of
  observables''.
\newblock \href{https://dx.doi.org/10.1103/PhysRevA.75.012328}{Phys. Rev. A
  {\bf 75}, 012328}~(2007).

\bibitem{WWLN10}
Hefeng Wang, Lian-Ao Wu, Yu-xi Liu, and Franco Nori.
\newblock ``Measurement-based quantum phase estimation algorithm for finding
  eigenvalues of non-unitary matrices''.
\newblock \href{https://dx.doi.org/10.1103/PhysRevA.82.062303}{Phys. Rev. A
  {\bf 82}, 062303}~(2010).

\bibitem{DGK14}
Anmer Daskin, Ananth Grama, and Sabre Kais.
\newblock ``A universal quantum circuit scheme for finding complex
  eigenvalues''.
\newblock \href{https://dx.doi.org/10.1007/s11128-013-0654-1}{Quantum Inf.
  Process. {\bf 13}, 333--353}~(2014).

\bibitem{DK17}
Ammar Daskin and Sabre Kais.
\newblock ``An ancilla-based quantum simulation framework for non-unitary
  matrices''.
\newblock \href{https://dx.doi.org/10.1007/s11128-016-1452-3}{Quantum Inf.
  Process. {\bf 16}, 1--17}~(2017).

\bibitem{Sha22}
Changpeng Shao.
\newblock ``Computing eigenvalues of diagonalizable matrices on a quantum
  computer''.
\newblock \href{https://dx.doi.org/10.1145/3527845}{ACM Transactions on Quantum
  Computing {\bf 3}, 1--20}~(2022).

\bibitem{SL22}
Changpeng Shao and Jin-Peng Liu.
\newblock ``Solving generalized eigenvalue problems by ordinary differential
  equations on a quantum computer''.
\newblock \href{https://dx.doi.org/10.1098/rspa.2021.0797}{Proc. R. Soc. A.
  {\bf 478}, 20210797}~(2022).

\bibitem{LS24}
Guang~Hao Low and Yuan Su.
\newblock ``Quantum eigenvalue processing''.
\newblock In 2024 IEEE 65th Annual Symposium on Foundations of Computer Science
  (FOCS).
\newblock \href{https://dx.doi.org/10.1109/FOCS61266.2024.00070}{Pages
  1051--1062}.
\newblock ~(2024).

\bibitem{Ber14}
Dominic~W Berry.
\newblock ``High-order quantum algorithm for solving linear differential
  equations''.
\newblock \href{https://dx.doi.org/10.1088/1751-8113/47/10/105301}{J. Phys. A
  Math. Theor. {\bf 47}, 105301}~(2014).

\bibitem{TOSU21}
Souichi Takahira, Asuka Ohashi, Tomohiro Sogabe, and Tsuyoshi~Sasaki Usuda.
\newblock ``Quantum algorithms based on the block-encoding framework for matrix
  functions by contour integrals''.
\newblock \href{https://dx.doi.org/10.26421/QIC22.11-12-4}{Quant. Inf. Comput.,
  {\bf 22}, 965--979}~(2021).

\bibitem{TOSU22}
Souichi Takahira, Asuka Ohashi, Tomohiro Sogabe, and Tsuyoshi~Sasaki Usuda.
\newblock ``Quantum algorithm for matrix functions by {C}auchy's integral
  formula''.
\newblock \href{https://dx.doi.org/10.26421/QIC20.1-2-2}{Quant. Inf. Comput.,
  {\bf 20}, 14--36}~(2022).

\bibitem{CAS+22}
Pedro~C.S. Costa, Dong An, Yuval~R. Sanders, Yuan Su, Ryan Babbush, and
  Dominic~W. Berry.
\newblock ``Optimal scaling quantum linear-systems solver via discrete
  adiabatic theorem''.
\newblock \href{https://dx.doi.org/10.1103/PRXQuantum.3.040303}{PRX Quantum
  {\bf 3}, 040303}~(2022).

\bibitem{Kro23}
Hari Krovi.
\newblock ``Improved quantum algorithms for linear and nonlinear differential
  equations''.
\newblock \href{https://dx.doi.org/10.22331/q-2023-02-02-913}{{Quantum} {\bf
  7}, 913}~(2023).

\bibitem{zhang2024exponential}
Xiao-Ming Zhang, Yukun Zhang, Wenhao He, and Xiao Yuan.
\newblock ``Exponential quantum advantages for practical non-{H}ermitian
  eigenproblems''~(2024).
\newblock  \href{http://arxiv.org/abs/2401.12091}{arXiv:2401.12091}.

\bibitem{ZZHY24}
Xiao-Ming Zhang, Yukun Zhang, Wenhao He, and Xiao Yuan.
\newblock ``Exponential quantum advantages for practical non-hermitian
  eigenproblems''~(2024).
\newblock  \href{http://arxiv.org/abs/2401.12091}{arXiv:2401.12091}.

\bibitem{TE05}
Lloyd~N. Trefethen and Mark Embree.
\newblock ``{Spectra and Pseudospectra: The Behavior of Nonnormal Matrices and
  Operators}''.
\newblock \href{https://dx.doi.org/10.1515/9780691213101}{Princeton University
  Press}. New Jersey~(2005).

\bibitem{BB98}
Carl~M. Bender and Stefan Boettcher.
\newblock ``Real spectra in {non-Hermitian Hamiltonians} having
  $\mathscr{P}\mathscr{T}$ symmetry''.
\newblock \href{https://dx.doi.org/10.1103/PhysRevLett.80.5243}{Phys. Rev.
  Lett. {\bf 80}, 5243--5246}~(1998).

\bibitem{AKS22}
Abhijeet Alase, Salini Karuvade, and Carlo~Maria Scandolo.
\newblock ``The operational foundations of {PT}-symmetric and quasi-{H}ermitian
  quantum theory''.
\newblock \href{https://dx.doi.org/10.1088/1751-8121/ac6d2d}{J. Phys. A: Math.
  Theor. {\bf 55}, 244003}~(2022).

\bibitem{Kar22}
Salini Karuvade.
\newblock ``Power and certifiability of quantum computing for open systems''.
\newblock \href{https://dx.doi.org/10.11575/PRISM/dspace/41251}{PhD thesis}.
\newblock University of Calgary.
\newblock Calgary, Alberta~(2022).

\bibitem{Mos02}
Ali Mostafazadeh.
\newblock ``{Pseudo-Hermiticity versus PT-symmetry III: Equivalence of
  pseudo-Hermiticity and the presence of antilinear symmetries}''.
\newblock \href{https://dx.doi.org/10.1063/1.1489072}{J. Math. Phys. {\bf 43},
  3944--3951}~(2002).

\bibitem{KAS22}
Salini Karuvade, Abhijeet Alase, and Barry~C. Sanders.
\newblock ``Observing a changing {H}ilbert-space inner product''.
\newblock \href{https://dx.doi.org/10.1103/PhysRevResearch.4.013016}{Phys. Rev.
  Research {\bf 4}, 013016}~(2022).

\bibitem{BBJ03}
Carl~M. Bender, Dorje~C. Brody, and Hugh~F. Jones.
\newblock ``{Must a Hamiltonian be Hermitian?}''.
\newblock \href{https://dx.doi.org/10.1119/1.1574043}{Am. J. Phys. {\bf 71},
  1095--1102}~(2003).

\bibitem{hughes2020calculus}
Deborah Hughes-Hallett, Andrew~M Gleason, and William~G McCallum.
\newblock ``Calculus: Single and multivariable''.
\newblock John Wiley \& Sons. ~(2020).

\bibitem{GSLW19}
Andr\'{a}s Gily\'{e}n, Yuan Su, Guang~Hao Low, and Nathan Wiebe.
\newblock ``Quantum singular value transformation and beyond: {E}xponential
  improvements for quantum matrix arithmetics''.
\newblock In Proceedings of the 51st Annual ACM SIGACT Symposium on Theory of
  Computing.
\newblock \href{https://dx.doi.org/10.1145/3313276.3316366}{Pages 193--204}.
\newblock New York, NY, USA~(2019). Association for Computing Machinery.

\bibitem{JL23}
Shi Jin and Nana Liu.
\newblock ``Quantum simulation of discrete linear dynamical systems and simple
  iterative methods in linear algebra via {S}chrodingerisation''.
\newblock \href{https://dx.doi.org/10.1098/rspa.2023.0370}{Proc. R. Soc. A.
  {\bf 480}, 20230370}~(2023).

\bibitem{ANB+22}
Abhijeet Alase, Robert~R. Nerem, Mohsen Bagherimehrab, Peter H\o{}yer, and
  Barry~C. Sanders.
\newblock ``Tight bound for estimating expectation values from a system of
  linear equations''.
\newblock \href{https://dx.doi.org/10.1103/PhysRevResearch.4.023237}{Phys. Rev.
  Res. {\bf 4}, 023237}~(2022).

\end{thebibliography}

\end{document}